\documentclass[12pt]{article}
\usepackage{fullpage,amsmath,amssymb,mathtools,natbib,hyperref}
\usepackage{titling}
\usepackage[skip=0pt]{caption} 
\usepackage{datetime,bigints}
\usdate
\usepackage{multirow} 
\usepackage{float}
\usepackage{algorithm,bbm}
\usepackage[]{algpseudocode}
\usepackage{arydshln}
\usepackage[section]{placeins}
\usepackage[titletoc,toc,title]{appendix}

\usepackage{xr,caption}

\renewenvironment{thebibliography}[1]{
  \begin{oldthebibliography}{#1}
    \setlength{\itemsep}{0em}
    \setlength{\parskip}{0em}
}
{
  \end{oldthebibliography}
}

\usepackage{thmtools} 
\usepackage{amsthm}
{
      \theoremstyle{plain}
      
      \newtheorem{theorem}{Theorem}
      \newtheorem{example}{Example}
      \newtheorem{proposition}{Proposition}
      \newtheorem{lemma}{Lemma}
      
      \newtheorem{assumption}{Assumption}
      
  }

\def\lQ{\scalebox{-1}[1]{''}}

\makeatletter
\renewenvironment{abstract}{%
    \if@twocolumn
      \section*{\abstractname}%
    \else 
      \begin{center}%
        {\bfseries \normalsize\abstractname\vspace{\z@}}
      \end{center} \vspace{-0.5cm}%
      \quotation
    \fi}
    {\if@twocolumn\else\endquotation\fi}
\makeatother
\usepackage[most,breakable]{tcolorbox}
\tcbset{enhanced jigsaw,colback=blue!5!white,colframe=blue!25!black,fonttitle=\bfseries}

\begin{document}

  \title{A Scrambled Method of Moments}
  \author{\Large Jean-Jacques Forneron\thanks{Department of Economics, Boston University, 270 Bay State Road, Boston, MA 02215 USA.\newline Email: \href{mailto:jjmf@bu.edu}{jjmf@bu.edu}, Website: \href{http://jjforneron.com}{http://jjforneron.com}.   \newline I would like to thank Zhongjun Qu for his helpful comments and suggestions. All errors are my own.}
  } \date{\today}
  \maketitle 

  \begin{abstract}  
    \vspace{-0.2cm}
    Quasi-Monte Carlo (qMC) methods are a powerful alternative to  classical Monte-Carlo (MC) integration. Under certain conditions, they can approximate the desired integral at a faster rate than the usual Central Limit Theorem, resulting in more accurate estimates. This paper explores these methods in a  simulation-based estimation setting with an emphasis on the scramble of \citet{owen1995randomly}. For cross-sections and short-panels, the resulting Scrambled Method of Moments simply replaces the random number generator with the scramble (available in most softwares) to reduce simulation noise. Scrambled Indirect Inference estimation is also considered. For time series, qMC may not apply directly because of a curse of dimensionality on the time dimension. A simple algorithm and a class of moments which circumvent this issue are described. Asymptotic results are given for each algorithm.  Monte-Carlo examples illustrate these results in finite samples, including an income process with \lQ lots of heterogeneity.'' 
  \end{abstract}
  

  \bibliographystyle{ecca}
  \baselineskip=18.0pt
  \thispagestyle{empty}
  \setcounter{page}{0}
  
  \newpage

  \section{Introduction}
  Simulation-based estimation is a popular approach to estimate complex economic models. The econometrician simply matches sample with simulated moments, drawn from a model of interest. The resulting Simulated Method of Moments (SMM) or Indirect Inference estimator makes estimation feasible even though the likelihood or the moments' expectation, required for MLE and GMM, may be impossible or impractical to compute. However, using simulations rather than analytical computations introduces simulation noise, which increases the variance of the estimates. Also, the resulting simulated objective function is typically non-smooth and hence more difficult to minimize numerically. In theory, simulating many samples reduces simulation noise and smoothes the objective function, making optimization easier. In practice, however, this may not be feasible because of the increased computational cost. Also, more informative moments used for Indirect Inference or the Efficient Method of Moments can to be more computationally demanding  than simpler moments used in SMM, leading to a computational tradeoff between the informativeness of the moments and simulation noise. 
  A practical solution is to use variance reduction techniques such as antithetic draws\footnote{See Section \ref{sec:MC_anti} for a brief overview of antithetic sampling.}  which can reduce simulation noise with nearly no computational overhead.

  This paper investigates quasi-Monte Carlo (qMC) integration, another variance reduction approach, in the context of simulation-based estimation, with an emphasis on the scramble of \citet{owen1995randomly}. Under certain conditions, qMC can approximate an expectation at a faster rate than the usual Monte-Carlo (MC) Central Limit Theorem (CLT). This suggests that a Scrambled Method of Moments could outperform conventional SMM estimates using as many simulated samples. This is shown to be the case for a large class of models in cross-sections and short-panels with potentially non-smooth moments as in \citet{McFadden1989,Pakes1989} or auxiliary parameters as in \citet{Gourieroux1993}. For time series, qMC may not apply directly because of a curse of dimensionality over the time dimension. A class of models and moments which circumvent this issue are described.
  
  Using the scramble in an estimation setting poses several practical and theoretical challenges. These sequences are designed to approximate a fixed integral of an \textit{iid} sequence. Improper use of the scramble under dependence or with covariates may result in inconsistent estimators. Hence, the first and main contribution of the paper is methodological. The second contribution is theoretical. Uniform Laws of Large Numbers (ULLN) and CLTs are provided to handle smooth moments in cross-sections and short-panels. Scrambled draws are random and identically distributed \textit{but not independent}. This makes it more challenging to handle time series and non-smooth moments. A stochastic equicontinuity result for cross-sections and short-panels is established by re-writing the scrambled empirical process as the sum of a non-identically distributed but independent array with a standard qMC sequence. This allows to invoke existing results for each term separately. In the time series setting, a similar strategy allows to invoke results for bounded dependent heterogeneous arrays.

  The finite sample properties of the Scrambled Method of Moments are illustrated using several simple Monte-Carlo examples including an income process with \lQ lots of heterogeneity'' \citep{BROWNING2010}. 
  In this example, the scramble improves on SMM with random and antithetic draws in terms of variance. Furthermore, optimization over the 2,000 replications  was on average 15\% faster with the scramble than SMM using as many simulated samples - because scrambled moments are smoother than MC moments.

  \subsection*{Structure of the Paper} 
  After a review of the literature, Section \ref{sec:MC_qMC} provides an overview of (quasi)-Monte Carlo integration which is lesser known in economics. Section \ref{sec:srambledMM} shows how to implement the Scrambled Method of Moments in various settings.  Asymptotic results for each algorithm are given in Section \ref{sec:AsymTheory} and the proofs are in Appendix \ref{sec:proofs}. Section \ref{sec:MC} illustrates its finite sample properties using Monte-Carlo simulations. Section \ref{sec:Conclusion} concludes.

  \subsection*{Related Literatures} 
  There are two related literatures: simulation-based estimation and variance reduction techniques. In economics, simulation-based estimation includes the Simulated Method of Moments \citep{McFadden1989,Pakes1989,duffie-singleton}, Indirect Inference \citep{Gourieroux1993} and the Efficient Method of Moments \citep{Gallant1996}. See \citet{smith-palgrave} for an overview of simulation-based estimation in economics and common empirical applications. In statistics, Bayesian methods such as Approximate Bayesian Computation \citep[also known as ABC; ][]{Marin2012} and Synthetic Likelihood \citep{Wood2010} are more common. See \citet{jjng-14} for an overview and comparisons of these frequentist and Bayesian methods. 

  As discussed in the introduction, Monte-Carlo methods introduce simulation noise in the estimation which increases the variance of the estimator. There is a large number of variance reduction techniques, the following summary will only cover some of those that are most relevant to simulation-based estimation. One approach is to use low-discrepancy sequences - this is more commonly known as quasi-Monte Carlo integration. These sequences were initially designed to compute integrals of \textit{iid} sequences and can achieve faster than $\sqrt{n}$-rate convergence. More details are given in Section \ref{sec:MC_qMC}. qMC integration has been extended to non-linear state-space filtering \citep{Gerber2015,Gerber2017}, MCMC sampling \citep{Owen2005} and importance sampling for ABC estimation \citep{Buchholz2017}. A key takeaway from these papers is that a lot of care is required in implementing qMC integration in non \textit{iid} settings (MCMC or filtering) where `naive' implementations may be inconsistent. This may explain why it is only rarely used in empirical economics, even though their appeal has been known for some time \citep[][]{judd1998numerical}.  
  In economics, antithetic draws are a popular variance reduction technique. However, they can lead to either efficiency gains or losses depending on the integrand as discussed in Section \ref{sec:MC_anti}. Another variance reduction method, which is more popular in statistics, is the control variates approach \citep[see e.g.][]{robert2013monte}.\footnote{Note that despite the similarity in names, this is not related to control variable estimation used in structural econometric estimation.} The main idea is to augment the estimating sample and simulated moments with analytically tractable moments for the shocks themselves. This additional information can help reduce the uncertainty attributable to simulation noise.\footnote{See \citet{Davis2019} for an application of control variates to Indirect Inference. Control variates were also considered for qMC integration  in \citet{Hickernell2005a}.} The control functional approach \citep{Oates2017}, which uses all the information about the distribution of the shocks, can result in faster than $\sqrt{n}$-rate convergence. Important efficiency gains require the control variate moments to be sufficiently rich which could lead to a curse of dimensionality. For instance, the model of Section \ref{sec:het_income} has shocks with dimension $d=30$ so that spanning polynomials of order up to $2$ or $3$ would require introducing $496$ or $5,456$ additional moments respectively. The number of moments quickly becomes greater than the sample size itself.

  \section{(quasi)-Monte Carlo Integration and the Scramble} \label{sec:MC_qMC}
  The following provides a brief overview of Monte-Carlo (MC) and quasi-Monte Carlo (qMC) integration.\footnote{For further reading, \citet{Lemieux2009} provides a non-technical introduction to MC and qMC integration; \citet{Dick2010} provide the underlying theory.}
  Throughout, we are interested in evaluating the integral of a known measurable function $f: [0,1]^d \to \mathbb{R}$:
  \begin{align}
    I = \int_{[0,1]^d} f(u)du, \label{eq:integral}
  \end{align}
  by using a fixed or random sequence of points $u_1,\dots,u_n$ in $[0,1]^d$:
  \begin{align} \hat I_n = \frac{1}{n} \sum_{i=1}^n f(u_i). \label{eq:sample_integral}
  \end{align}
  
  \subsection{Monte-Carlo Integration and Antithetic Draws} \label{sec:MC_anti}
  A widely applicable approach is MC integration. Take \textit{iid} uniform draws $u_i \sim \mathcal{U}_{[0,1]^d}$ and compute the sample analog $\hat I^{MC}_n = \frac{1}{n}\sum_{i=1}^n f(u_i).$
  Assuming $f(u_i)$ has finite variance, $\hat I^{MC}_n$ is unbiased and the approximation error $|\hat I^{MC}_n-I|$ is of order $\sqrt{\text{var}[f(u_i)]/n}$. This implies that in order to reduce the approximation error tenfold, the number of draws must be a hundred times greater: the computational cost increases faster than the approximation error declines.
  
  A popular variance reduction approach is to use antithetic draws. For $n$ even, compute: \[ \hat I_n^{Anti} = \frac{1}{n} \sum_{i=1}^{n/2}[f(u_i) + f(1-u_i)], u_i \overset{iid}{\sim} \mathcal{U}_{[0,1]^d}.\] This approach is only valid if $f(u_i)$ and $f(1-u_i)$ have the same distribution; for instance, $e_i = \Phi^{-1}(u_i) \sim \mathcal{N}(0,1)$ and $-e_i = \Phi^{-1}(1-u_i) \sim \mathcal{N}(0,1)$ as well. Without this property, when the distribution is asymmetric, $\hat I_n^{Anti}$ may not be consistent for $I$. 
  
  Assuming $f(u_i)$ and $f(1-u_i)$ have the same distribution, $\hat I_n^{Anti}$ is unbiased and $\text{var}(\hat I_n^{Anti}) = \left( \text{var}[f(u_i)] + \text{cov}[f(u_i),f(1-u_i)] \right)/n$. If $\text{corr}[f(u_i),f(1-u_i)] = -1$, then $\text{var}(\hat I_n^{Anti}) =0$; the estimator is exact as soon as $n=2$. This improves significantly on MC integration. However, if $\text{corr}[f(u_i),f(1-u_i)] = +1$ then $\text{var}(\hat I_n^{Anti}) = 2 \text{var}(\hat I_n^{MC})$. Now, $\hat I_n^{MC}$ outperforms $\hat I_n^{Anti}$.
  
  The performance of antithetic draws relative to simple MC draws will typically depend on both the parameter of interest and the choice of  estimating moments.  
  To illustrate, consider the following two examples. First, suppose  $I=\mathbb{E}(e_i)$ where $e_i = \Phi^{-1}(u_i) \sim \mathcal{N}(0,1)$. Note that $-e_i \sim \mathcal{N}(0,1)$ and $\hat I_n^{Anti}$ is consistent for $I$ in this example. Since $\text{corr}(e_i,-e_i) = -1$, $\text{var}(\hat I_n^{Anti})=0$. The estimator is exact as soon as $n=2$. Second, suppose $I=\mathbb{E}(e_i^2)$ with $e_i$ as above. Now, $\text{corr}( e_i^2,[-e_i]^2 )=+1$ and $\text{var}(\hat I_n^{Anti}) = 2 \text{var}(\hat I_n^{MC})$. These examples suggest that the moments need to have some asymmetry properties in order to produce efficiency gains. This can be hard to check for intractable non-linear models.

  \subsection{quasi-Monte Carlo Integration} \label{sec:qMC}
  The discussion above shows that some sequences can outperform MC integration. For instance, for $f$ smooth and $u_i \in [0,1]^d$ with $d=1$ the lattice sequence $u_i = i/(n-1), i=0,\dots,n-1$, the estimator
  \[ \hat I_n^{Lattice} = \frac{1}{n} \sum_{i=1}^{n}f(u_i), \quad u_i = i/(n-1),\, i \in \{0,\dots,n-1\}\]
  has an approximation error of order $O(\|\partial_u f\|_{\infty}/n)$. The approximation error declines linearly with the computational cost. However for $d\geq 2$, this sequence has approximation errors of order $n^{-1/d}$ which is worse than MC as soon as $d \geq 3$. 
   
  It is possible to break this curse of dimensionality. To achieve this, the qMC literature relies on two pivotal inequalities. The first one is the Koksma-Hlawka inequality:
  \begin{align} 
    \Big| \frac{1}{n} \sum_{i=1}^n f(u_i) - \int_{[0,1]^d} f(u)du \Big| \leq \|f\|_{TV} \times D_n^\star(u_1,\dots,u_n), \label{eq:KHineq}
  \end{align}
  where $\|f\|_{TV}$ is the total variance norm of $f$ in the sense of Hardy and Krause: 
  \begin{align}  \|f\|_{TV} = \sum_{\mathfrak{u} \subseteq \mathcal{I}_d}\int_{[0,1]^{|\mathfrak{u}|}} \Big|\frac{\partial^{|\mathfrak{u}|} f (\mathfrak{u})}{\partial \mathfrak{u}}\Big|d\mathfrak{u}, \label{eq:HardyKraus}
  \end{align}
  $\partial^{|\mathfrak{u}|} f (\mathfrak{u})/\partial \mathfrak{u}$ consists of all univariate derivatives $\partial_{u_1} f(u),\dots,\partial_{u_d} f(u)$ and partial cross-derivatives $\partial^2_{u_1,u_2} f(u),\partial^2_{u_1,u_3} f(u),\partial^2_{u_2,u_3} f(u),\dots,\partial^2_{u_{d-1},u_d} f(u)$ up to order $d$ with $\partial^d_{u_1,\dots,u_d} f(u)$. It \textit{does not include repeated derivatives} such as $\partial^2_{u_1,u_1} f(u)$. What matters here is the smoothness of $f$ across the co-ordinates $u_1,\dots,u_d$. As a result, integrating over larger dimensions $d$ typically requires additional smoothness in $f$ over these cross-derivatives.
  
  The other term in the Koksma-Hlawka inequality is $D_n^\star(u_1,\dots,u_n)$ which corresponds to the star discrepancy of the sequence $(u_1,\dots,u_n)$, defined as: 
  \begin{align} D_n^\star(u_1,\dots,u_n) = \sup_{u \in [0,1)^d} \Big| \frac{1}{n}\sum_{i=1}^n \mathbbm{1}_{u_i \in [0,u)} - \int_{[0,u)} 1 du\Big|. \label{eq:star_discrepancy}
  \end{align}
  In statistics, this is known as the Kolmogorov-Smirnov (KS) distance between the empirical CDF of $(u_1,\dots,u_n)$ and the population CDF of a uniform $\mathcal{U}_{[0,1]^d}$ distribution. 

  For a given function $f$, reducing the approximation error in (\ref{eq:KHineq}) implies finding sequences with smaller $D_n^\star$. For \textit{iid} random draws, $D_n^* = O_p(n^{-1/2})$ by Donsker theorems \citep{VanderVaart1996}. The lattice sequence above has $D_n^* = O(n^{-1/d})$ for $d \geq 1$.
  
  For any sequence $(u_1,\dots,u_n)$, the second pivotal inequality - initially due to \citet{Roth1954} and generalized by \citet{Schmidt1970} - provides a lower bound on its star discrepancy:
  \begin{align} D_n^\star(u_1,\dots,u_n) \geq C_d \times \frac{\log(n)^{d-1}}{n}, \label{eq:lower_bound}
  \end{align}
  where $C_d$ is a universal constant which only depends on the dimension $d$. Note the striking difference with the Discrepancy of the set $D_n(u_1,\dots,u_n) = \sup_{i = 1,\dots,n}(\inf_{j \neq i}\|u_i-u_j\|)$ which cannot decrease faster than $n^{-1/d}$. Under the sup-norm distance there is the well known curse of dimensionality which affects grid searches, non-parametric estimation, etc. Under the KS distance, this lower bound suggests that the impact of dimensionality is much less severe.\\
  
  \begin{tcolorbox}[breakable,title=Constructing a qMC Sequence: the Sobol Point Set] The following material is adapted from \citet{Dick2010}, Chapter 8.1, and \citet{Lemieux2009}, Chapter 5.4. A popular approach to conduct qMC integration is to use sequences called \textit{Digital Nets}. Many of these sequences can be represented as:
  \[ u_i = \sum_{j=0}^\infty u_{i,j}b^{-j},  \]
  where $b \geq 2$ is a prime number so that $(u_{i,0},u_{i,1},\dots)$ is the $b$-adic representation $u_i$, i.e. the digits of $u_i$ in the basis $b$.\footnote{In base $10$, the digits are simply the number's decimals.}
  A well known digital net is the Sobol sequence for which $b=2$ so that $(u_{i,0},u_{i,1},\dots)$ is simply the digital expansion of $u_i$ in base $2$. The following considers the case $d=1$ for simplicity. To construct the sequence two inputs are needed. First, we need \textit{primitive polynomials} sorted by increasing degree $e_\ell, \ell \in \{1,\dots,d\}$:\footnote{A primitive polynomial is a polynomial of degree $e \geq 1$ with coefficients in a Galois field $\mathbb{G}(m) = \mathbb{Z} \text{ modulo } m \times \mathbb{Z}$ (for instance $\mathbb{G}(2) = \mathbb{Z} \text{ modulo } 2 \times \mathbb{Z} = \{-1,0,1\}$; a field is a finite set where addition, subtraction, multiplication and division are defined and verify certain axioms; in addition, a Galois field is finite) such that the powers $x^j \text{ modulo } p(x), j=1,\dots,m^{e}-1$ generate the set of nonzero polynomials of degree less of equal to $e$ in $\mathbb{G}(m)$.}
  \[ p_\ell(x) = x^{e_\ell} + a_{\ell,1} x^{e_\ell-1} + \dots + a_{\ell,e_\ell-1}x + a_{\ell,e_\ell},\quad \ell \in \{1,\dots,d\}.\]
  Second, construct \textit{direction numbers} $v^\ell_{j}, \ell \in \{1,\dots,d\}, j \in \{1,\dots,e_\ell\}$ as:
  \[ v_{j}^\ell = 2^{-\ell}m_{j}^\ell \]
  for some user-chosen odd integers $m^\ell_{j} \in \{1,\dots,2^j-1\}, j\in \{1,\dots,e_\ell\}$. One way to think of these direction numbers is that $2^{-\ell}$ splits $[0,1]$  into subintervals of length $2^{-\ell}$ and $m^\ell_{j}$ picks one of these subintervals (see the figure below for an illustration). Then the recursions described below ensure that the sequence covers $[0,1]^d$ well using these subintervals (this is a defining feature of digital nets). From these initial direction numbers, the following recursion generates the rest of the sequence:
  \[ v^\ell_{j+1} = a_{\ell,1}v^\ell_{j} \oplus \dots \oplus a_{\ell,e_\ell-1} v^\ell_{j+2-e_\ell} \oplus v^\ell_{j+1-e_\ell} \oplus (2^{-e_\ell}v^\ell_{j+1-e_\ell}),\]

  where $\oplus$ is the x-or operator on the binary representation.\footnote{The x-or or exclusive-or operator has the following property: $1\oplus 1 =0, 1 \oplus 0 = 1, 0 \oplus 0 =0$. On the binary representation this implies for $v_1 = 1/2, v_2 = 3/4$ we have $v_1 =  1 \times 2^{-1}, v_2 = 1 \times 2^{-1} + 1 \times 2^{-2}$ so that $v_1 \oplus v_2 = (1,0) \oplus (1,1) = (0,1)$ which is $1/4$ in the usual decimal representation.} The x-or operator allows to cycle over the splits described above and the requirement that the polynomial be primitive ensures that the cycle spans all the splits.
  Now to compute the i-th Sobol number, write down the base $2$ representation $i = i_0 + 2i_1+2^2i_2 + \dots + 2^{r-1}i_r$ for some $r \geq 0$ and we have:
  \[ u^\ell_{i} = i_0v^\ell_{1} \oplus  i_1v^\ell_{2} \oplus \dots \oplus i_{r-1}v^\ell_{r},\]
  the Sobol point set is then $u_i = (u^1_{i},\dots,u^d_{i}), i \in \{1,\dots,n\}$. \\
  Note that since the initial direction numbers are user-chosen many Sobol sequences can be generated with varying finite sample properties. One issue in particular is that some direction numbers can lead to finite sample correlations between the dimensions of $u_i$ which is undesirable. Several authors report direction numbers which perform well in practice; some scrambling algorithms can also improve the properties of the sequence \citep[see e.g.][]{Chi2005}. In practice, the Fortran implementation of ACM Algorithm 659\citep{Bratley1988,Joe2003} seems to be widely used \footnote{The \textit{randtoolbox} package of \citet{Dutang2019} provides an R interface to the Fortran code.} and provides direction numbers with good properties for dimensions up to $d=1,111$.
  \tcbline
  To put this in practice, consider the case with $d=1$, $p(x) = x^2 + x + 1$ so that $e=2$, which means that two direction numbers are required. Pick $v_{1}=1/2,v_2=3/4$ or in binary representation $v_1 = (1,0),v_2=(1,1)$. Using the recursion:
  \[ v_3 = (1,1) \oplus (1,0) \oplus (0,0,1) = (0,1,1),\]
  since the polynomial coefficients are all equal to $1$ and $2^{-1}v_1=(0,0,1)$. Note that $v_3 = 2^{-2} + 2^{-3} = 1/4 + 1/8 = 3/8 = 0.375$. The next number in the sequence is:
  \[ v_4 = (0,1,1) \oplus (1,1) \oplus (0,0,1,1) = (1,0,0,1) \]
  since $2^{-2}v_2 = (0,0,1,1)$, in base $10$ we have $v_4= 2^{-1} + 2^{-4} = 0.5625$. The Sobol sequence is then $u_0 = 0$, for $i=1 = 1 \times 2^0$, we have $u_1 = v_1 = (0,1)$ which implies $v_1 = 1/2$ in the decimal system. For $i=2 = 0 \times 2^0 + 1 \times 2^1$, we have $u_2 = (0 \times v_1) \oplus (1 \times v_2) = v_2$ i.e. $u_2= 3/4$, $u_3 = v_1 \otimes v_2 = (0,1)$ which is $1/4$, then $u_4 = v_3$ i.e. $3/8$. The figure below illustrates the construction of the first $8$ points of the Sobol sequence by the R package \textit{randtoolbox}.
  \begin{center}
    \includegraphics[scale = 0.7]{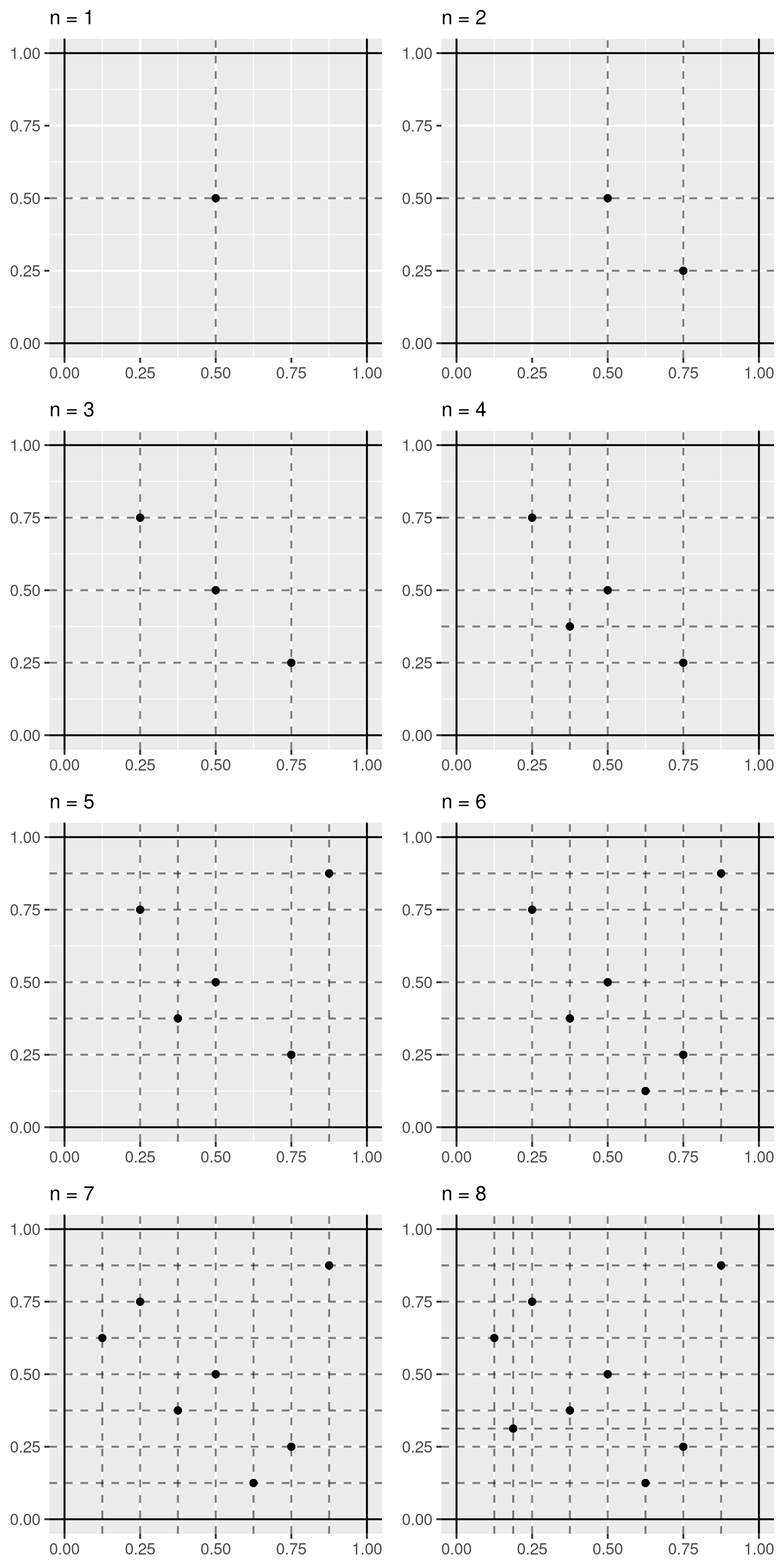}
  \end{center}
  \end{tcolorbox}  

  For any $d \geq 1$ fixed, this lower bound suggests a faster than $\sqrt{n}$-rate is achievable. There are a number of deterministic sequences which are close in rate to the bound (\ref{eq:lower_bound}); these include the Sobol, Halton, van der Corput and Hammersley sequences. Most of these are readily available in statistical softwares.\footnote{The R package \textit{randtoolbox}, the \textit{Sobol} module in Julia, the \textit{quasirandomset} toolset in Matlab and the \textit{SamplePack} library in C++ can generate the Sobol sequence, for instance.} There is, however, a caveat that when $d$ becomes large $C_d$ can also become large. For instance, $C_d = 2^d$ for the Sobol sequence which increases very rapidly with $d$. As a result, in finite samples MC integration may outperform qMC integration for $d$ large relative to $n$. Under additional smoothness conditions, so-called higher-order sequences can achieve even faster rates of order $n^{-\alpha}$ (up to log-terms) for some $\alpha>1$ which depends on the smoothness of $f$ and higher-order properties of the sequence.
  
  \subsection{Randomized quasi-Monte Carlo and the Scramble} \label{sec:rqMC}
  These results above are restrictive since they require the integrand $f$ to be smooth, otherwise $\|f\|_{TV} = +\infty$ and the Koksma-Hlawka inequality (\ref{eq:KHineq}) is uninformative. In economics, many problems involve non-smooth integrands such as simulation-based estimation of discrete choice models \citep{Train2009}. Also, $\hat I_{n}^{qMC}$ computed with a deterministic sequence is typically biased and its approximation error is hard to evaluate numerically. This would make it difficult to compute standard errors in an estimation setting. 
  
  One solution is to use randomized quasi-Monte Carlo (RqMC) methods. A simple randomizer is the digital shift. Take one random draw $u \sim \mathcal{U}_{[0,1]^d}$, a qMC sequence $u_1,\dots,u_n$ ( e.g. Sobol, Halton) and compute $\tilde u_i = [u_i + u] \text{ modulo } 1$. The modulo operator is applied one dimension at a time. This shifts all the co-ordinates of $u_1,\dots,u_n$ by the same random quantity $u$ and preserves the order of magnitude of its star discrepancy $D_n^\star$. The randomized $\tilde u_i$ are identically distributed $\mathcal{U}_{[0,1]^d}$ but not independent. The estimator $\hat I_n^{RqMC}$ is unbiased. To approximate its variance, apply the digital shift with different draws $u$ to compute the integral several times and then compute the variance across these estimates \citep{Lemieux2009}.
  
  Another randomization approach, which will be the main focus of this paper, is the scramble introduced by \citet{owen1995randomly}. Similarly to the random shift above, it transforms a deterministic low-discrepancy sequence into random identically but not independently distributed uniform $\mathcal{U}_{[0,1]^d}$ draws. Since the scrambled draws are uniform, the estimator $\hat I_n^{scramble}$ is unbiased. The procedure is described in the box below. The scramble does not deteriorate the discrepancy of the original sequence, in fact it was shown that it can further improve it \citep[see][Chapter 13.1 for bibliographical references]{Dick2010}. 
  
  The scramble approximates $I$ under the same conditions as the classical CLT as shown in Theorem \ref{th:owen}. The underlying theory is quite involved since it relies of Walsh expansions, an approach similar to Fourier expansions but in a digital basis $b$ which requires an understanding of both number theory and functional approximation theory. See \citet{Dick2010} for an introduction to the relevant material and proofs.   
  \begin{theorem}[Owen, 1997] \label{th:owen} Let $u_1,\dots,u_n$ be a scrambled sequence using the algorithm proposed by \citet{owen1995randomly}. If $f$ is measurable and $f(u), u \sim \mathcal{U}_{[0,1]^d}$ has finite variance then:
    \[ \frac{1}{n} \sum_{i=1}^n f(u_i) - \int_{[0,1]^d} f(u)du = o_p(n^{-1/2}). \]
  \end{theorem}
  Under additional smoothness conditions $\hat I_n^{scramble}$ approximates $I$ at a near $n^{-3/2}$-rate; which is faster than deterministic qMC sequences. A refinement of the orginal algorithm, higher-order scrambling, can achieve even faster rate for smooth integrands; for instance, in some cases the convergence can be of order $n^{-5/2}$ or $n^{-7/2}$. 
  The scrambled estimator $\hat I_n^{Scramble}$ is unbiased. Its variance can be approximated the same way as for $\hat I_n^{RqMC}$. Note that these results assume $d \geq 1$ is fixed. In practice, MC may outperform the scramble for $d$ large. Other scrambles have also been proposed by \citet{Hickernell1996} and \citet{Matousek1998}, among others. See \citet{Lemieux2009} and \citet{Dick2010} for additional references.\\

  \begin{tcolorbox}[breakable,title=Owen's Scramble]
    The following material is adapted from \citet{Owen1997}. 
    The scramble starts from a set of points $u_i = (u_{i}^1,\dots,u_{i}^d) \in [0,1]^d$ with base $b$ representation: 
    \[ u_i^\ell = u_{i,1}^\ell b^{-1} + u_{i,2}^\ell b^{-2} + \dots,\]
    for each coordinate $\ell \in \{1,\dots,d\}$. The scrambled points $\tilde u_i$ are generated by applying random permutations to the $b$-adic representation $u_{i,j}^\ell$ of $u_i$. Let $\pi^\ell$ be random permutations from $\{0,\dots,b-1\}$ to itself drawn uniformly over all permutations and independently across coordinates $\ell$ (there are $b! = b \times (b-1) \times (b-2) \times \dots \times 1$ such permutations), then the scrambled sequence is generated recursively as:
  \begin{align*} 
    \tilde u_{i,1}^\ell = \pi^\ell(u_{i,1}^\ell),  \quad
    \tilde u_{i,2}^\ell = \pi^\ell_{u_{i,1}^\ell}(u_{i,2}^\ell), \quad
    \tilde u_{i,3}^\ell = \pi^\ell_{u_{i,1}^\ell,u_{i,2}^\ell}(u_{i,3}^\ell), \quad
    \dots
  \end{align*}
  the permutation for the $j$-th digit depends on the $j-1$ previous digits; this creates path dependence in the scrambling process which makes the algorithm computationally demanding.   
  Owen's algorithm, described above, is also known as nested uniform or fully random scrambling. ACM Algorithm 823 implements a faster non-nested scrambling algorithm (which relies on matrix operations) that is also called Owen's scramble in statistical softwares \citep{Hong2003,Dutang2019}.\footnote{As discussed in \citet{Hong2003}, this is ``to recognize that it is done in the spirit of Owen’s original proposal.''} Although the two implementations are different, the resulting sequences share important desirable theoretical properties.
  \tcbline
  To illustrate the nested scramble described above, consider a Sobol sequence written in base $b=2$ with $d=1$. There are two possible permutations: $\pi(0)=1,\pi(1)=0$ and $\pi(0)=0,\pi(1)=1$. First, the permutation is applied to $u_{i,1}$. The permutation $\pi(0)=0$ preserves the first digit. In practice, this implies that $u_i \geq 0.5 \Rightarrow \tilde u_i \geq 0.5$. The other possible permutation $\pi(1)=0$ splits the $[0,1)$ segment into two parts $[0,1/2)$ and $[1/2,1)$ and permutes them: $u_i \geq 0.5 \Rightarrow \tilde u_i < 0.5$. 
  
  The second step permutes the second digit: split the unit interval into $4$ subintervals $[0,1/4), [1/4,1/2), [1/2,3/4), [3/4,1)$ and apply a permutation as before but between the pairs $[0,1/4), [1/4,1/2)$ and $[1/2,3/4), [3/4,1)$.  
  For instance, suppose that $\pi(0)=0$ so that the first digit is unchanged. Consider the pair $[0,1/4), [1/4,1/2)$, and assume the permutation is $\pi_{0}(0)=1$ then $u_i \in [0,1/4) \Rightarrow \tilde u_i \in [1/4,1/2)$. Separately, if $\pi_{1}(0)=0$ then $u_i \in [1/2,3/4) \Rightarrow \tilde u_i \in [1/2,3/4)$.
  
  The third step further splits $[0,1)$ into $8$ subintervals $[0,1/8),[1/8,1/4),[1/4,3/8),[3/8,1/2),\dots$ and applies permutations over the $4$ pairs following the same logic. For instance if $u_i \in [0,1/8)$, then $\tilde u_{i,3} = \pi_{00}(0)$ while $\tilde u_{i,3} = \pi_{01}(0)$ when $u_i \in [1/4,3/8)$. Note that $\pi_{00}$ and $\pi_{01}$ are different uniform permutations draws. The process continues until a desired level of precision is attained. The table below illustrates the first two iterations of the scramble when applied to a small sequence with $d=1$ for some realization of the permutations.\\
  \begin{center}
    \includegraphics[scale = 0.7]{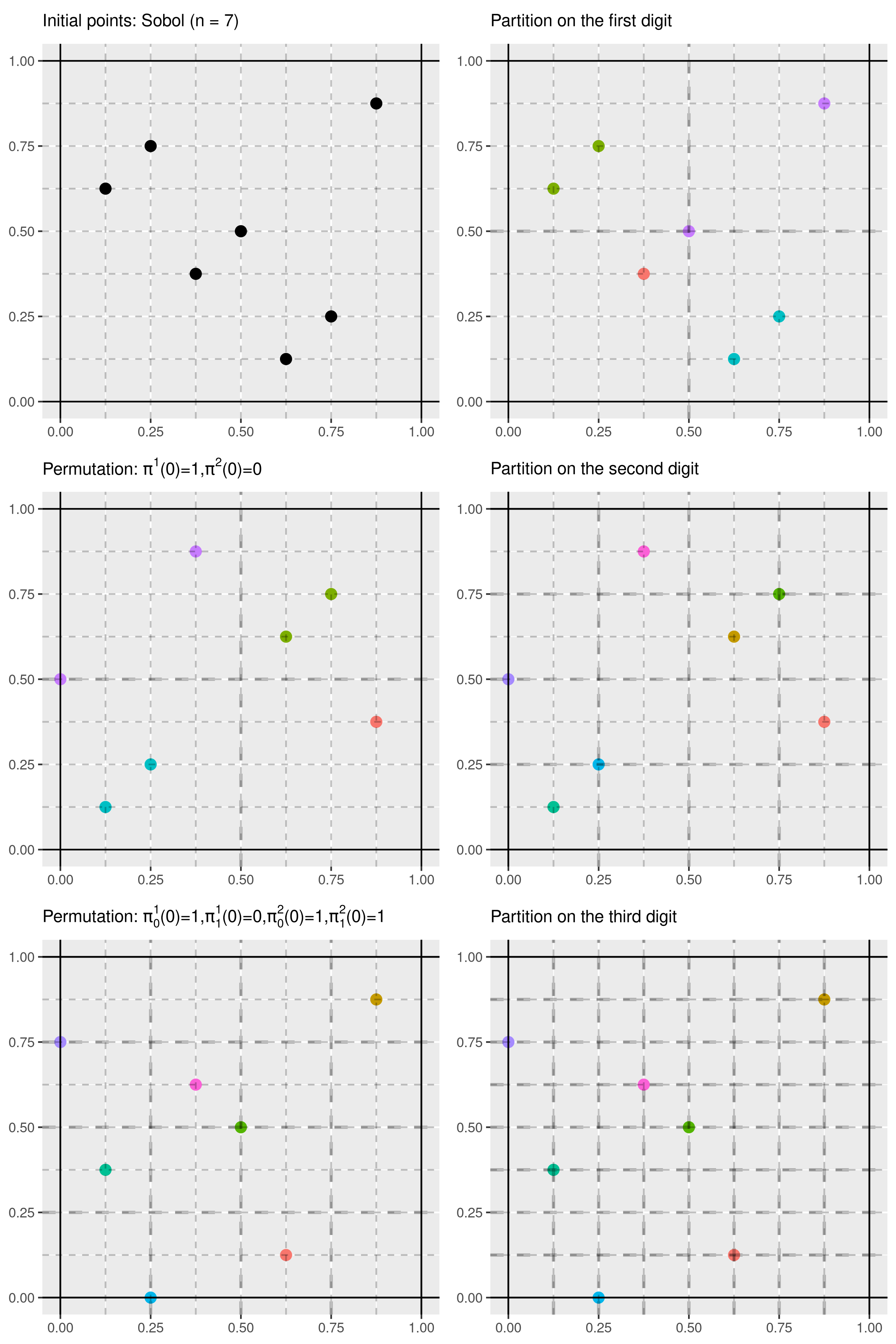}
  \end{center}
    \begin{center}
    \begin{tabular}{l|cccc}
      \hline \hline
     & $u_1$ & $u_2$ & $u_3$ & $u_4$ \\ 
      \hline 
      initial points          & 0.125 & 0.375 & 0.500 & 0.875 \\ 
      $\pi(0)=1$              & 0.625 & 0.875 & 0.000 & 0.375 \\ 
      $\pi_0(0)=0,\pi_1(0)=1$ & 0.625 & 0.875 & 0.500 & 0.125 \\
       \hline \hline
    \end{tabular}
  \end{center}
   The figure above illustrates the first two iterations of the nested scramble applied the Sobol sequence with $n=7$ and $d=2$. The first iteration splits $[0,1)^2$ into 4 squares and performs permutations over the two rectangles on the $x$-axis ($\pi^1$) and $y$-axis ($\pi^2$). The second iteration further splits each square into 4 sub-squares (so there is a total of 16 squares) and performs permutations between the 2 pairs of rectangles on the $x$-axis ($\pi_0^1$ and $\pi_1^1$) and the $y$-axis ($\pi_0^2$ and $\pi_1^2$). The next iteration further splits each square into 4 subsets and performs additional permutations. The procedure continues until a certain level of numerical precision is achieved. 
   
   Note that although the $\tilde u_i^\ell$ are not independent over $i$ for a given $\ell$, they are independent over $\ell$ for any $i$ since the permutations are drawn independently over dimensions $\ell \in \{1,\dots,d\}$. This feature is quite important for the finite sample properties of the scramble: while the Sobol sequence could display correlations across dimensions $\ell$ for some direction numbers, the nested scramble guarantees independence over $\ell$. This is visible in Figure \ref{fig:draws} where some Sobol points are aligned on the 45 degree line whereas the scrambled Sobol sequence does not display such patterns.
  \end{tcolorbox}

  \begin{figure}[H]
    \caption{Random, Sobol and Scrambled Sequences: $n=500,d=2$} \label{fig:draws}
    \begin{center}
    \includegraphics[scale=0.9]{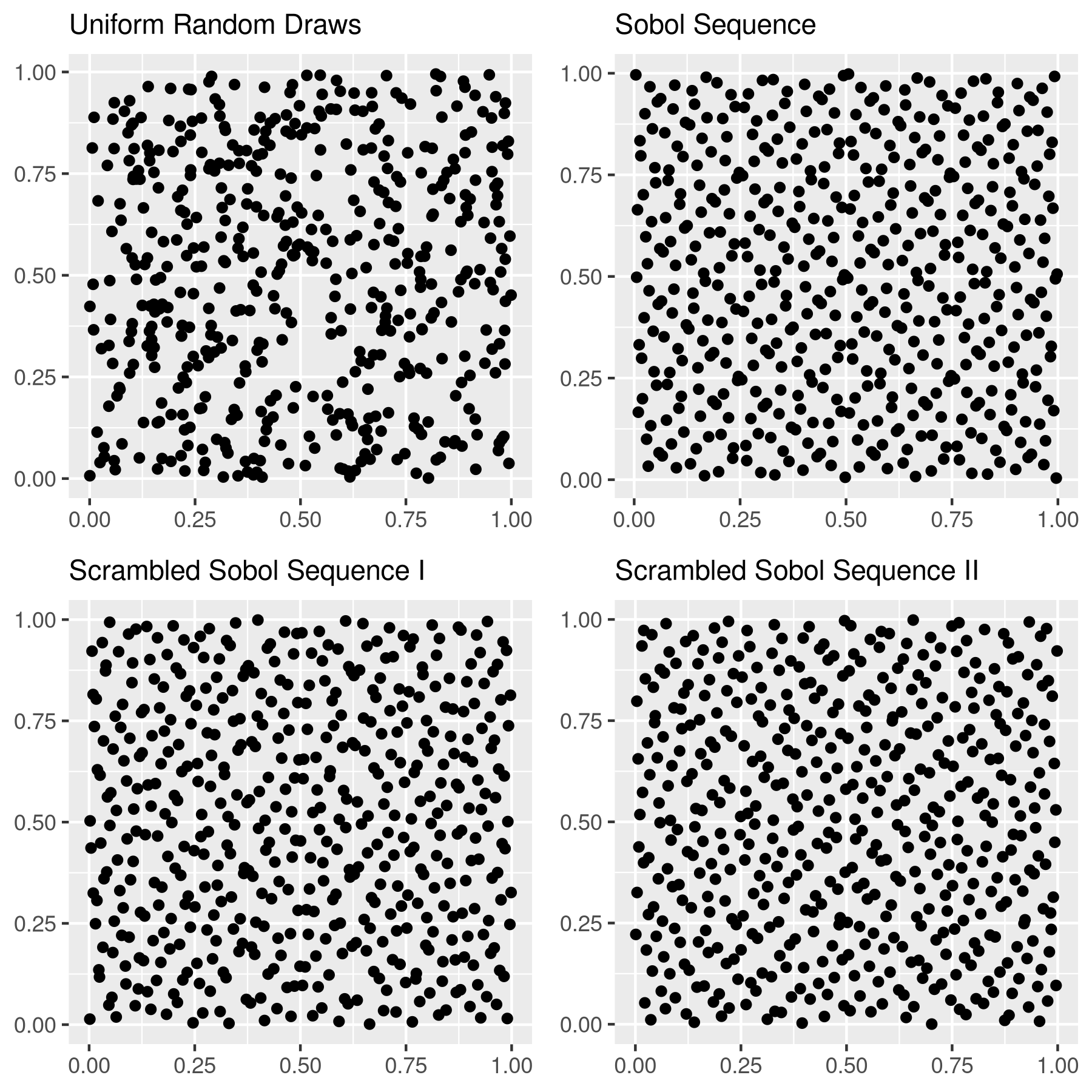}
    \end{center}
    \textbf{Legend:} top-left random uniform draws, top-right Sobol sequence, bottom-left one realization of the scrambled Sobol sequence,  bottom-right another realization of the scrambled Sobol sequence. 
  \end{figure}

  Figure \ref{fig:draws} illustrates the differences between random, deterministic and scrambled sequences. A particular realization of a random sample may have points clustered in some areas of $[0,1]^2$, as visible in the figure. The Sobol sequence covers the square better in this case, though some points might appear to cluster here as well. The two realizations of the scrambled Sobol sequence cover the square well and cluster slightly less (better discrepancy) than the deterministic Sobol points.  
 
  \section{A Scrambled Method of Moments} \label{sec:srambledMM}
  This section introduces the main algorithms to implement the Scrambled Method of Moments and Scrambled Indirect Inference. The data generating process (DGP) is the same as in \citet{Gourieroux1993}:
  \begin{align}
      &y_{i,t} = g_{obs}(y_{i,t-1},x_{i,t},z_{i,t};\theta) \label{eq:obs}\\ 
      &z_{i,t} = g_{latent}(z_{i,t-1},u_{i,t};\theta) \text{ where } u_{i,t} \overset{iid}{\sim} \mathcal{U}_{[0,1]^d}. \label{eq:latent}
  \end{align}
 A simple transformation allows to replace $u_{i,t}$ with $e_{i,t} = \Phi^{-1}(u_{i,t}) \overset{iid}{\sim} \mathcal{N}(0,1)$ or other distributions by the Rosenblatt transform. $i=1,\dots,n$ indexes individuals and $t=1,\dots,T$ the time dimension. $y_{i,t}$ is the vector of observed outcome variables. $x_{i,t}$ is a vector of strictly exogenous covariates and $z_{i,t}$ a vector of unobserved latent variables. The functions $g_{obs}$ and $g_{latent}$ are assumed to be known up to a finite dimensional parameter $\theta$ to be estimated. 

  \subsection{Static Models} \label{sec:SMM_static}
  For static models, which correspond to cross-sections and short-panels, the $t$ index will be omitted to re-write (\ref{eq:obs})-(\ref{eq:latent}) and the moments as:
  \begin{align}
    &y_i = g(x_i,u_i;\theta), \quad \hat \psi_n = \frac{1}{n} \sum_{i=1}^n \psi(y_i,x_i), \label{eq:static}
  \end{align}
  where $y_i = (y_{i,1},\dots,y_{i,T})$ and $u_{i} = (u_{i,1},\dots,u_{i,T}) \in [0,1]^{T \times d_{u_{i,t}}}$. The dimension $d$ of $u_i$ is $T \times \text{dim}(u_{i,t})$, using the notation in (\ref{eq:obs})-(\ref{eq:latent}).
  Given a vector of moments $\hat \psi_n$ and a weighting matrix $W_n$, a simple SMM estimator is given in Algorithm \ref{algo:SMM}.
 \begin{algorithm}
  \caption{Simulated Method of Moments for Static Models} \label{algo:SMM}
  \begin{algorithmic}
    \State \textbf{Draw a random sequence} $u^s_{i} \overset{iid}{\sim} \mathcal{U}_{[0,1]^d}, i = 1,\dots,n$; and $s = 1, \dots, S$
    \State Simulate:  $y_{i}^s(\theta) = g_{obs}(x_{i},u_i^s;\theta)$
    \State Compute: $\hat \psi_n^S(\theta) = \frac{1}{n \times S} \sum_{s=1}^S \sum_{i=1}^n \psi(y_{i}^s(\theta),x_{i})$
    \State Find: $\hat \theta_n^S = \text{argmin}_{\theta \in \Theta} \|\hat \psi_n-\hat \psi_n^S(\theta)\|_{W_n}$
  \end{algorithmic}
\end{algorithm}
    
Without covariates $x_i$, the expectation $\mathbb{E}[\hat \psi_n^S(\theta)]$ has the same form as (\ref{eq:integral}). The scramble can be applied if the moments have finite variance. The resulting Algorithm \ref{algo:static_no_covariates} is thus very similar to SMM. 
\begin{algorithm}
  \caption{Scrambled Method of Moments for Static Models without Covariates} \label{algo:static_no_covariates}
  \begin{algorithmic}
    \State \textbf{Draw a scrambled sequence}  $\tilde u_{i} \overset{iid}{\sim} \mathcal{U}_{[0,1]^d}, i = 1,\dots,n \times S$ 
    \State Simulate:  $\tilde y_{i}(\theta) = g_{obs}(\tilde u_i;\theta)$
    \State Compute: $\hat \psi_n^S(\theta) = \frac{1}{n \times S}\sum_{i=1}^{n \times S} \psi( \tilde y_{i}(\theta))$
    \State Find: $\hat \theta_n^S = \text{argmin}_{\theta \in \Theta} \|\hat \psi_n-\hat \psi_n^S(\theta)\|_{W_n}$
  \end{algorithmic}
\end{algorithm}
In practice, one samples an $(nS) \times d$ matrix of scrambled shocks rather than $S$ different $n \times d$ matrices of random numbers. This is may be useful because using a large simulated sample of $n \times S$ observations implies a reduction in variance greater than $S$, as a consequence of the faster rate in Theorem \ref{th:owen}, compared to using $S$ independent simulated samples. Asymptotic results for Algorithm \ref{algo:static_no_covariates} are provided in Proposition \ref{prop:static_nocovariates}, assuming the moments are smooth in $\theta$. These assumptions are comparable to those required for SMM.

When $\hat \psi_n$ is a vector of auxiliary moments as \citep{Gourieroux1993}, the results from Proposition \ref{prop:static_nocovariates} can be extended for the scramble as shown in Proposition \ref{prop:ind}. Again, the assumptions are comparable to those required for Indirect Inference. These Indirect Inference results could also be extended to non-smooth moments and time series given appropriate changes to the assumptions.

In the presence of covariates, $\mathbb{E}[\hat \psi_n^S(\theta)]$ does not have the same form as (\ref{eq:integral}):
\[ \mathbb{E}[\hat \psi_n^S(\theta)] = \int_{[0,1]^d \times \mathcal{X}}  \psi\left( g(x,u;\theta),x\right) f_x(x)dx du, \]
where $f_x$ is joint density of the covariates $x$.\footnote{The results could be extended to allow some components of $x$ to be discrete. However, the assumptions in the next Section imply that at least one of the covariates should have a continuous density.} Without further assumptions, it is typically not possible to sample from the population $f_x$ directly so that qMC sequence with $n \times S$ elements for $(x_i,u_i)$ cannot be constructed. Taking the covariates as given, Algorithm \ref{algo:static_covariates} relies on $S$ independent scrambled sequences of size $n$ rather than a large sequence of size $n \times S$ as in Algorithm \ref{algo:static_no_covariates}.\footnote{It is implicitly assumed that $(x_1,\dots,x_n)$ is a random sample. If the ordering is deterministic, then $x_i$ and $u_i)$ are not independent. Randomly shuffling the covariates without replacement solves this issue.}
    \begin{algorithm}
      \caption{Scrambled Method of Moments for Static Models with Covariates} \label{algo:static_covariates}
  \begin{algorithmic}
      \State \textbf{Draw $S$ independently scrambled sequences }  $\tilde u_{i}^s \overset{iid}{\sim} \mathcal{U}_{[0,1]^d}, i = 1,\dots,n$
      \State Simulate:  $\tilde y^s_{i}(\theta) = g_{obs}(x_{i}, \tilde u^s_i;\theta)$
      \State Compute: $\hat \psi_n^S(\theta) = \frac{1}{n \times S}\sum_{i=1}^{n \times S} \psi( \tilde y^s_{i}(\theta),x_i)$
      \State Find: $\hat \theta_n^S = \text{argmin}_{\theta \in \Theta} \|\hat \psi_n-\hat \psi_n^S(\theta)\|_{W_n}$
  \end{algorithmic}
    \end{algorithm}

  For each $s \in \{1,\dots,S\}$, the function $\mathbb{E}[ \psi( y^s_i(\theta),x_i)|\tilde u_i^s = u]$ does not depend on $x$ so that Theorem \ref{th:owen} can be applied to this conditional expectation, assuming it has finite variance. This insight was used to derive CLTs for moments based on hybrid sequences which combine MC draws with qMC sequences in \citet{Okten2006} and \citet{Buchholz2017} for bounded $\psi$. The results in Proposition \ref{prop:static_covariates} extend these results to unbounded empirical processes over $\theta \in \Theta$, allowing $\hat \psi_n^s$ to be non-smooth in $\theta$. The assumptions are more demanding than for SMM, although they could be weakened for smooth moments with covariates. The conditional expectation $\mathbb{E}[\hat \psi_n^S(\cdot)|\tilde u_1,\tilde u_2,\dots]$ itself is required to be smooth in $\theta$, \textit{i.e.} integrating out the covariates smoothes out the sample and simulated moments. This implies that at least one of the covariates has a continuous density.

  \subsection{Dynamic Models} \label{sec:SMM_dynamic}
    For dynamic models, which correspond to time series observations, the $i$ index will be omitted to re-write (\ref{eq:obs})-(\ref{eq:latent}) and the moments as:
  \begin{align}
    &y_t = g_{obs}(y_{t-1},z_t;\theta), \quad z_t = g_{latent}(z_{t-1},u_t;\theta), \quad u_t \overset{iid}{\sim} \mathcal{U}_{[0,1]^d} \label{eq:dyn_model}
  \\
    &\hat \psi_T = \frac{1}{T} \sum_{t=L+1}^T \psi(y_t,\dots,y_{t-L}). \label{eq:dyn_mom}
  \end{align}
  Covariates are omitted to simplify the theoretical results. 
  Only moments involving a fixed and finite number of lags $L$ will be considered as explained below. Algorithm \ref{algo:SMMdyn} details the SMM procedure to estimate (\ref{eq:dyn_model})-(\ref{eq:dyn_mom}). 
  \begin{algorithm}
    \caption{Simulated Method of Moments for Dynamic Models} \label{algo:SMMdyn}
    \begin{algorithmic}
      \State \textbf{Draw a random sequence} $u^s_t \overset{iid}{\sim} \mathcal{U}_{[0,1]^d}, t = 1,\dots,T;$ $s = 1, \dots, S$
      \State Set $(y_0^s,z_0^s) = (y_0,z_0)$, a fixed initial value
      \State Simulate: $z_t^s(\theta) = g_{latent}(z_{t-1}^s,u_t^s;\theta)$ and $y_t^s(\theta) = g_{obs}(y_{t-1}^s(\theta),z_t^s(\theta);\theta)$
      \State Compute: $\hat \psi_T^S(\theta) = \frac{1}{T \times S} \sum_{s=1}^S \sum_{t=L+1}^n \psi(y_t^s(\theta),\dots,y_{t-L}^s(\theta))$
      \State Find: $\hat \theta_T^S = \text{argmin}_{\theta \in \Theta} \|\hat \psi_T-\hat \psi_T^S(\theta)\|_{W_T}$
    \end{algorithmic}
  \end{algorithm}

  To understand the issues caused by the dynamics for the scramble and qMC integration, note that for any initial value $(y_0,z_0)$, $y_t$ can be re-written as:
  \[ y_t = g_t(u_t,\dots,u_1,y_0,z_0;\theta), \]
  for some function $g_t$ which can be expressed in terms of $g_{obs}$ and $g_{latent}$. Using this notation, the expected value of $\hat \psi_T$ can be re-written as:
  \[ \mathbb{E}(\hat \psi_T) = \frac{1}{T} \sum_{t=L+1}^T \int_{[0,1]^{t \times d}} \psi \circ (g_t,\dots,g_{t-L})(u_t,\dots,u_1,y_0,z_0) du_t\dots du_1.  \]
  The expectation above differs from the qMC setting in (\ref{eq:integral}) in several ways. First, the function to be integrated involves $g_t$ which varies with $t$ unlike the function in (\ref{eq:integral}). Second, the integral is computed over $u_1,\dots,u_t$ which has a dimension $t$ that increases with the sample size. This implies a curse of dimensionality for qMC which requires the dimension $d$ to be fixed. Third, both randomized and non-randomized qMC sequences are identically but not independently distributed. A naive implementation of the scramble could introduce spurious dependence in the simulated data and the resulting estimator may not be consistent as a result.
  
  Implementing qMC integration in a dynamic setting without additional structure comes at a cost. In finance, qMC sequences are used to simulate long time series and price financial derivatives \citep[see e.g.][]{Paskov1995,Lemieux2009}. This is done by setting $d = T$ and sampling a very large number $n$ of financial series. In the present setting, this amounts to picking $S$ very large and $d=T$ which is not computationally attractive compared to standard SMM.\footnote{Recall that for the Sobol sequence $C_d = 2^d$ so that the error would be of the order of $2^T/S$. Consistency of the qMC integral would require $S \gg 2^T$, \textit{i.e.} $S$ needs to grow exponentially fast with the sample size $T$.} For state-space filtering, \citet{Gerber2015,Gerber2017} propose a Hilbert sorting step to re-sample draws into a low-discrepancy sequence using the Hilbert fractal map from $[0,1]$ to $[0,1]^d$. This Hilbert map can be challenging to implement in practice and suffers from a curse of dimensionality. 
  
  \subsubsection{qMC-only Approach} \label{sec:SMM_dyn_qMC}
  The class of moments described in (\ref{eq:dyn_mom}) where the number of lags $L$ is fixed and finite allows to circumvent these issues. To get some intuition, suppose that it is possible to draw $(y^1_t,z^1_t) = F_{y,z}^{(-1)}(v^1_t)$ from the stationary distribution directly using the Rosenblatt transform with $v^1_t \sim \mathcal{U}_{[0,1]^{\text{dim}(v_t^1)}}$. Then, using additional shocks $u_t^2,\dots,u_t^L$ one could simulate a short time series consisting of $L \geq 1$ observations for each $t = 1,\dots,T \times S$:
  \begin{align*}
    &(y^1_t,z^1_t) = F_{y,z}^{(-1)}( v^1_t)\\
    &(y_t^2,z_t^2) = \left(g_{obs}(y_t^1,z_t^2;\theta),\, g_{latent}(z_t^1, u_t^2;\theta) \right)\\
    &\quad \vdots\\
    &(y_t^L,z_t^L) = \left(g_{obs}(y_t^{L-1},z_t^L;\theta),\, g_{latent}(z_t^{L-1}, u_t^{L};\theta) \right).
  \end{align*}
  The resulting draws $(y_t^1,\dots,y_t^L)$ are \textit{iid} over $t =1,\dots,T\times S$ from the stationary distribution by construction.\footnote{This idea was also used in \citet{Davis2019} but as a variance reduction method with MC draws.} This is now within the setting of (\ref{eq:integral}). Algorithm \ref{algo:ScrMMdyn1} describes a Scrambled Method of Moments for models where simulating as described above is feasible. The main idea is to simulate the $(y_t^1,\dots,y_t^L)$ $T \times S$ times with scrambled shocks $(v_t^1,u_t^2,\dots,u_t^L)_{t = 1,\dots,T\times S} \in [0,1]^d$ with  dimension $d=\text{dim}(v_t^1,u_t^2,\dots,u_t^L)$. which depends on the dimension of the shocks and the numbers of lags $L$. Note that while Algorithm \ref{algo:SMMdyn} requires $T \times S$ draws, Algorithm \ref{algo:ScrMMdyn1} effectively requires $n \times S \times L$ draws. However, the latter Algorithm is massively parallel over $t$ so that for some models it may run faster than the former in a parallel environment. Proposition \ref{prop:dyn_case1} provides the asymptotic results for Algorithm \ref{algo:ScrMMdyn1}.\footnote{When simulating the initial draw with the Rosenblatt transform is not possible, one may consider using a fixed starting value and a \textit{burn-in} period assuming some decay conditions hold. This is only considered for the hybrid MC-qMC method, theoretical investigations for qMC-only draws is left to future research.}
    \begin{algorithm}
      \caption{Scrambled Method of Moments for Dynamic Models - qMC-only Approach} \label{algo:ScrMMdyn1}
      \begin{algorithmic}
        \State \textbf{Draw a scrambled sequence}  $\tilde u_t = (\tilde v_{t}, \tilde u_t^2,\dots,\tilde u_{t}^L) \in [0,1]^{d \times (L-1) + \tilde d}$, $t = 1,\dots,T \times S$
        \State Compute $(\tilde y^1_t(\theta),\tilde z^1_t(\theta)) = F^{-1}(\tilde v;\theta)$ for $t =1, \dots, T\times S$
        \State Simulate: $\tilde z^\ell_t(\theta) = g_{latent}(\tilde z_{t}^{\ell-1},\tilde u_t^\ell;\theta)$ and $\tilde y_t^\ell(\theta) = g_{obs}(\tilde y_{t}^{\ell-1}(\theta),\tilde z_t^\ell(\theta);\theta)$ for $\ell =2,\dots, L$
        \State Compute: $\hat \psi_T^S(\theta) = \frac{1}{T \times S} \sum_{t=1}^{T\times S}  \psi(\tilde y_t^L(\theta),\dots,\tilde y_t^1(\theta))$
        \State Find: $\hat \theta_T^S = \text{argmin}_{\theta \in \Theta} \|\hat \psi_T-\hat \psi_T^S(\theta)\|_{W_T}$
      \end{algorithmic}
    \end{algorithm}

    Sampling from the stationary distribution directly is feasible for some DGPs such as the Gaussian ARMA model (see the Monte-Carlo example in Section \ref{sec:arma}) or the following stochastic volatility process:
    \[ \log(\sigma_t) = \mu_\sigma + \rho_\sigma \log(\sigma_{t-1}) + \kappa_\sigma e_{t,1}, \quad y_t = \sigma_t e_{t,2},\quad (e_{t,1},e_{t,2}) \overset{iid}{\sim} \mathcal{N}(0,I_2).  \]
    Since the log-volatility follows a Gaussian AR(1) process, one can simply draw $\log(\sigma_t^1) \sim \mathcal{N}(\mu_\sigma/(1-\rho_{\sigma},\kappa_\sigma^2/[1-\rho_\sigma^2])$ and $y_t^1 = \sigma^1_t e^1_{t,2}$ where $e^1_{t,2}$ to simulate $(y_t^1,\sigma_t^1)$ from their stationary distribution. For more complex DGPs this may not be feasible, however.

    \subsubsection{Hybrid MC-qMC Approach}
    When the direct approach in Algorithm \ref{algo:ScrMMdyn1} is not feasible, an alternative is to sample the initial draws $(y_t^1,z_t^1)$ by MC methods and then simulate $(y_t^2,z_t^2),\dots,(y_t^L,z_t^L)$ using the scramble. This hybrid MC-qMC approach allows to sample from intractable distributions while retaining some of the the features of qMC integration.

    \begin{algorithm}
      \caption{Scrambled Method of Moments for Dynamic Models - Hybrid MC-qMC Approach} \label{algo:ScrMMdyn2}
      \begin{algorithmic}
        \State \textbf{Draw a random sequence} $u^1_t \overset{iid}{\sim} \mathcal{U}_{[0,1]^d}, t = 1,\dots,T \times S$
        \State Set $(y_0^s,z_0^s) = (y_0,z_0)$, a fixed initial value
        \State Simulate: $z_t^1(\theta) = g_{latent}(z_{t-1}^1,u_t^1;\theta)$ and $y_t^1(\theta) = g_{obs}(y_{t-1}^1(\theta),z_t^1(\theta);\theta)$
        \State \textbf{Draw a scrambled sequence}  $\tilde u_t = (\tilde u_t^2,\dots,\tilde u_{t}^L) \in [0,1]^{d \times (L-1)}$, $t = 1,\dots,T \times S$
        \State Simulate: $\tilde z^\ell_t(\theta) = g_{latent}(\tilde z_{t}^{\ell-1},\tilde u_t^\ell;\theta)$ and $\tilde y_t^\ell(\theta) = g_{obs}(\tilde y_{t}^{\ell-1}(\theta),\tilde z_t^\ell(\theta);\theta)$ for $\ell =2,\dots, L$
        \State Compute: $\hat \psi_T^S(\theta) = \frac{1}{T \times S} \sum_{t=1}^{T\times S}  \psi(\tilde y_t^L(\theta),\dots,\tilde y_t^1(\theta))$
        \State Find: $\hat \theta_T^S = \text{argmin}_{\theta \in \Theta} \|\hat \psi_T-\hat \psi_T^S(\theta)\|_{W_T}$
      \end{algorithmic}
    \end{algorithm}
  The resulting Algorithm \ref{algo:ScrMMdyn2} combines elements from Algorithms \ref{algo:SMMdyn} and \ref{algo:ScrMMdyn1}. It      requires an additional loop compared to the latter, which is more computationally demanding. Because the estimation combines MC with qMC, the variance of the estimates will typically be greater than a qMC only approach. Note that once the $(z_t^1(\theta),y_t^1(\theta))$ are drawn by MC simulations, $(z_t^\ell(\theta),y_t^\ell(\theta))_{\ell > 1}$ can be simulated in parallel which can be computationally attractive. 
     
  Proposition \ref{prop:dyn_case2} provides asymptotic results for Algorithm \ref{algo:ScrMMdyn2} with conditions similar to \citet{duffie-singleton} but assuming bounded moments. Relaxing this assumption would require to extend existing CLTs for dependent heterogeneous arrays \citep[see e.g.][Theorme 5.10]{White1984} which goes beyond the scope of this paper. The simulations in Section \ref{sec:arma} suggest that the estimator performs well with unbounded moments in practice.

  \subsection{Computing Standard Errors for the Simulated and Scrambled Method of Moments} \label{sec:se_s}
  Given that the scramble is different from standard Monte-Carlo methods, the following shows how to compute standard errors for $\hat \theta_n^{S}$ for SMM, antithetic draws and the scramble. 

  Under regularity conditions, the Simulated and Scrambled Method of Moments estimators satisfy the following asymptotic expansion:
  \[ \hat \theta_n^{S} - \theta_0 = - \left( G^\prime W_n G \right)^{-1} G^\prime W_n \left[\hat \psi_n - \hat \psi_n^S(\theta_0) \right] + o_p(n^{-1/2}), \]
  where $G = \partial_\theta \mathbb{E} \left[ \hat \psi_n^S(\theta_0) \right]$ is the usual Jacobian matrix. Under a CLT, the asymptotic variance is given by the usual sandwich formula. Given that $W_n$ is chosen by the user, only two terms need to be approximated: the Jacobian $G$ and the asymptotic variance of $[\hat \psi_n - \hat \psi_n^S(\theta_0)]$.
  
  When the moments are smooth, the plug-in Jacobian estimator $\hat G_n = \partial_\theta  \hat \psi_n^S(\hat \theta_n^S)$ is consistent for $G$ under a ULLN. For non-smooth moments, there are several possibilities. The more computationally demanding approach is to Bootstrap the estimator $\hat \theta_n^S$ directly. Alternatively, \citet{Bruins2018} propose to smooth the draws $y_{i,t}^s$ in dynamic discrete choice models using a kernel; this transforms non-smooth and unbiased into smooth but biased simulated moments. \citet{Frazier2019} rely on a change of variable argument to compute analytical Jacobians in a class of discrete choice models. The quasi-Jacobian matrix in \citet{Forneron2019} smoothes the moments themselves to approximate $G$. It is also possible to use MCMC methods to sample from a quasi-posterior distributions which approximates the frequentist distribution of $\hat \theta_n^S$ \citep[see e.g.][]{chernozhukov-hong,Wood2010}.

  For cross-sections and short panels, the asymptotic variance of $[\hat \psi_n - \hat \psi_n^S(\theta_0)]$ in SMM can be approximated with the cross-sectional variance of $[\psi(y_i,x_i)-\frac{1}{S} \sum_{s=1}^S \psi(y_i^s(\hat \theta_n^S),x_i)]$. Pooling all the simulated samples that way ensures that the estimator is consistent for both standard and antithetic draws.\footnote{Another approach is to use use the variance of $\psi(y_i^s(\hat \theta_n^S),x_i)$ divided by $S$ as an estimate for $\hat \psi_n^S(\theta_0)$. Although commonly used, this may actually not be consistent in the presence of antithetic draws. Depending on the correlation described in Section \ref{sec:MC_anti} it may either under or over-estimate the variance.} For time series, under appropriate conditions, a HAC estimator is consistent for the long-run variance of $\hat \psi_T$ and the averaged $\hat \psi^S_T(\theta_n^S)$ respectively. Computing the long-run variance for the averaged $\sum_{s=1}^S \psi(y_t^s,\dots,y_{t-L}^s)/S$ ensures that the estimate is consistent for both standard and antithetic draws. As before, an estimate for the non-averaged moment may not be consistent for antithetic draws because of the dependence between simulated moments.

  For the Scrambled Method of Moments, the variance should not be computed as above because scrambled draws are not independent from one another. Theorem \ref{th:owen} implies that the asymptotic variance only involves $\hat \psi_n$ in most cases; because simulation noise is asymptotically negligible.\footnote{See e.g. Proposition \ref{prop:static_nocovariates}.} One approach is to only compute the variance of $\hat \psi_n$. However, as illustrated in Section \ref{sec:MC}, even though the simulation noise can be small in finite samples, it may not be completely negligible for some DGPs. In these cases, one would want to account for the variance attributable to $\hat \psi_n^S$. As discussed in Section \ref{sec:rqMC}, to consistently estimates the variance of $\hat \psi_n^{S}$ one can evaluate $\hat \psi_n^{S}$ several times with different seeds for the scramble and compute the variance across these estimates.

  \section{Asymptotic Theory} \label{sec:AsymTheory}
  In the following $\hat \theta_n^S$ and $\hat \theta_T^S$ will denote the scrambled estimator for static and dynamic models respectively. Consistency and asymptotic normality results are provided for the algorithms described above. The first set of assumptions below is standard in the Monte-Carlo simulation-based estimation literature.
  \begin{assumption}[Identification, Regularity, Sample Moments] \label{ass:regular}
    Suppose the following holds:
    \begin{itemize}
      \item[i.] (Identification) $\mathbb{E}[\hat \psi_n] = \mathbb{E}[\hat \psi_n^S(\theta)] \Leftrightarrow \theta = \theta_0$.
      \item[ii.] (Regularity) $\theta_0 \in \text{interior}(\Theta)$ where $\Theta$ is a compact and convex subset of $\mathbb{R}^{d_\theta}$, $1 \leq d_\theta<+\infty$ fixed. $\mathbb{E}[\hat \psi_n^S(\cdot)]$ is continuously differentiable around $\theta_0$ and $\partial_\theta \mathbb{E}[\hat \psi_n^S(\theta_0)]$ has full rank.
      \item[iii.] (Sample Moments) $\hat \psi_n$ satisfies a Law of Large Numbers and a Central Limit Theorem:
      \[ \sqrt{n} \left[  \hat \psi_n - \mathbb{E}(\hat \psi_n) \right]\overset{d}{\to} \mathcal{N}(0,V). \]  
      \item[iv.] (Weighting Matrix) $W_n \overset{p}{\to} W$ positive definite
    \end{itemize}
  \end{assumption} 
  
  \subsection{Static Models} \label{sec:Static}



  To simplify notation, let:
  \[ \tilde \psi(x_i,\theta,u_i) \overset{def}{=} \psi( g(x_i,\theta,u_i),x_i ) \]

  \subsubsection{Smooth moments with no covariates}

  \begin{assumption}[Scrambled Smooth Moments without Covariates] \label{ass:smooth_nocovariates}
    Suppose that the following holds:
    \begin{itemize}
      \item[i.] For all $\theta \in \Theta$, \[\mathbb{E}\Big( \Big\| \tilde \psi( \theta,u_i) \Big\|^{2} \Big) < + \infty\]
      and 
      \[ \|\tilde \psi( \theta_1,u_i) - \tilde \psi( \theta_2,u_i)\| \leq C_1(u_i) \times \|\theta_1-\theta_2\|, \]
      where $\mathbb{E}[C_1(u_i)^2]<+\infty$.
      \item[ii.] For all $\theta \in \Theta$, $\tilde \psi$ is continuously differentiable in $\theta$ around $\theta_0$ and:
      \[\mathbb{E}\Big( \Big\| \partial_\theta \tilde \psi(\theta,u_i) \Big\|^{2} \Big) < + \infty,\]
      and 
      \[ \| \partial_\theta \tilde \psi( \theta_1,u_i) - \partial_\theta \tilde \psi( \theta_2,u_i)\| \leq C_2(u_i) \times \|\theta_1-\theta_2\|, \]
      where $\mathbb{E}[C_2(u_i)^2]<+\infty$.
    \end{itemize} 
  \end{assumption}
  Assumption \ref{ass:smooth_nocovariates} provides sufficient conditions to prove a uniform law of large numbers (ULLN) for $\tilde \psi$ and $\partial_\theta \tilde \psi$ using the scramble. The proof is similar to \citet{Jennrich1969}.

  \begin{proposition}[Consistency and Asymptotic Normality without Covariates] \label{prop:static_nocovariates}
    Suppose Assumptions \ref{ass:regular} and \ref{ass:smooth_nocovariates} hold, then $\hat \theta_n^S \overset{p}{\to} \theta_0$ and
    \[ \sqrt{n} \left( \hat \theta_n^S - \theta_0 \right) \overset{d}{\to} \mathcal{N}(0,\Sigma), \]
    where
    \[ \Sigma = \left( G^\prime W G \right)^{-1} G^\prime W V W G \left( G^\prime W G \right)^{-1},\]
    $G = \partial_\theta \mathbb{E} [\hat \psi_n^S(\theta_0)]$, $V  =\lim_{n \to \infty} n \times \text{var}(\hat \psi_n)$.
  \end{proposition}
  Given the ULLN for the simulated moments and the fast convergence rate in Theorem \ref{th:owen}, the estimator is consistent and asymptotically normal. The main difference with standard SMM is that here the simulations do not inflate the asymptotic variance, even for $S=1$, whereas the simulation noise in SMM implies an additional $1/S$ factor. 

  \subsubsection{Potentially non-smooth moments with covariates}
  As discussed in Section \ref{sec:SMM_static}, moments with covariates do not quite fit the setting described in Section \ref{sec:MC_qMC}. Indeed, the scrambled draws are identically distributed but not independent. With the introduction of covariates, $\tilde \psi(x_i,u_i;\theta),i = 1,\dots,n$ are neither identically distributed nor independent which makes deriving ULLNs and CLTs challenging. Furthermore, if the moments are non-smooth in $\theta$ then the approach of \citet{Jennrich1969} cannot be applied and empirical process methods are required. 
  
  The main idea is to split the sample moments and the empirical process into two parts: one is non-identically distributed but independent and the other is identically distributed but not independent. The former can be handled using CTLs and empirical process results for heterogeneous arrays and assuming the later is smooth in $\theta$, it can be handled using the steps in \citet{Jennrich1969} as in Proposition \ref{prop:static_nocovariates}. The main assumption there is that integrating over $x_i$, while conditioning on $u_i$, transforms non-smooth into smooth moments. This puts restrictions on the moments and covariates used in the estimation. 
  \begin{assumption}[Scrambled Non-Smooth Moments with Covariates] \label{ass:smooth}
    Suppose that for some $\delta >0$ the following holds:
    \begin{itemize}
      \item[i.] $\mathbb{E} \left[ \text{var}\left( \psi(y_i,x_i) - \tilde \psi(x_i,u_i;\theta_0)  | u_i\right) \right]$ is positive definite and finite, also\\ $\mathbb{E}\left[ \| \text{var}\left( \psi(y_i,x_i) - \tilde \psi(x_i,u_i;\theta_0) | u_i \right) \|^{2+\delta} \right]<+\infty$. 
      \item[ii.] There exists an envelope function $\bar \psi$ such that for all $\theta \in \Theta$, $\| \tilde \psi(x_i,u_i;\theta) \| \leq \bar \psi(x_i,u_i)$ with $\mathbb{E} \left[ \text{var} \left( \bar \psi(x_i,u_i) | u_i \right) \right] >0$ and  $\mathbb{E} \left[ \text{var} \left( \bar \psi(x_i,u_i)|u_i \right)^{2+\delta} \right) < +\infty$.
      \item[iii.] There exists $\tilde C_1(\cdot)$ such that $\theta_1,\theta_2 \in \Theta$, $\mathbb{E}( \|\tilde \psi(x_i,u_i;\theta_1)-\tilde \psi(x_i,u_i;\theta_2)\|^2|u_i) \leq \tilde C(u_i)^2 \times \|\theta_1-\theta_2\|^2$ with $\mathbb{E}\left( \tilde C(u_i)^4 \right) < +\infty$.
      \item[iv.] $\mathbb{E}(\tilde \psi(x_i,u_i;\cdot)|u_i)$ is continuously differentiable in $\theta \in \Theta$, $u_i$ almost surely. There exists $\tilde C_2(\cdot)$ such that for all $\theta_1,\theta_2 \in \Theta$, $\| \mathbb{E} \left( \tilde \psi(x_i,u_i;\theta_1) - \tilde \psi(x_i,u_i;\theta_2) | u_i \right) - \partial_\theta \mathbb{E} \left( \tilde \psi(x_i,u_i;\theta_2) |u_i \right)(\theta_1-\theta_2) \| \leq \tilde C_2(u_i) \times \|\theta_1 - \theta_2\|^2$. There exists $\tilde C_3(\cdot)$ such that for all $\theta \in \Theta$, $\mathbb{E}\left[ \| \partial_\theta \mathbb{E} \left(  \tilde \psi(x_i,u_i;\theta) |u_i \right) \|^2 \right] <+\infty$, $\| \partial_\theta \mathbb{E} \left[ \tilde \psi( \theta_1,u_i)|u_i \right] - \partial_\theta \mathbb{E} \left[ \tilde \psi( \theta_2,u_i) |u_i \right]\| \leq \tilde C_3(u_i) \times \|\theta_1-\theta_2\|$,  where $\mathbb{E}[ \tilde C_3(u_i)^2]<+\infty$.
    \end{itemize} 
  \end{assumption}    
  Assumption \ref{ass:smooth} i.-ii. ensure the Lindeberg condition holds for the heterogeneous array which is required to apply a CLT and the Jain-Markus Theorem \citep{VanderVaart1996}. Conditions iii-iv. ensures that Theorem \ref{th:owen} can be applied to the smoothed moments, i.e. after integrating out the covariates. 

  \begin{proposition}[Consistency and Asymptotic Normality with Covariates] \label{prop:static_covariates} For $S \geq 1$, suppose that $\|\hat \psi_n - \hat \psi_n^S(\hat \theta_n^S)\|_{W_n} \leq o_p(n^{-1/2})$ and that Assumptions \ref{ass:regular}, \ref{ass:smooth} hold then $\hat \theta_n^S \overset{p}{\to} \theta_0$ and
    \[ \sqrt{n} \left( \hat \theta_n^S - \theta_0 \right) \overset{d}{\to} \mathcal{N}(0,\Sigma), \]
    where
    \[ \Sigma = \left( G^\prime W G \right)^{-1} G^\prime W \tilde V W G \left( G^\prime W G \right)^{-1},\]
    $G = \partial_\theta \mathbb{E} [\hat \psi_n^S(\theta_0)]$, $\tilde V  = \mathbb{E} \left[ \text{var} \left( \psi(y_i,x_i) - \tilde \psi(x_i,u_i;\theta_0)|u_i  \right) \right]$. 
  \end{proposition}
  Proposition \ref{prop:static_covariates} is similar to \citet{Pakes1989} with scrambled instead of MC draws. The variance $\tilde V$ can be computed using the steps described in Section \ref{sec:se_s}.

  \subsubsection{Scrambled Indirect Inference} \label{sec:II}
  The following extends the results from Proposition \ref{prop:static_nocovariates} to the Indirect Inference estimator of \citet{Gourieroux1993}. The moments $\hat \psi_n,\hat \psi_n^S(\theta)$ are now defined as sample and simulated M-estimators:
  \begin{align*}
    \hat \psi_n &= \text{argmin}_{\psi \in \Psi} M_n(\psi), \quad \quad \text{where } M_n(\psi) \,\,\,\, = \frac{1}{n} \sum_{i=1}^n m(y_i;\psi) \\
    \hat \psi_n^S(\theta) &= \text{argmin}_{\psi \in \Psi} M^S_n(\theta;\psi), \quad \text{where } M^S_n(\theta;\psi)=  \frac{1}{nS} \sum_{i=1}^{nS} m(y_i^s(\theta);\psi).
  \end{align*}
  Again, to simplify notation consider:
  \[ \tilde m(u_i,\theta;\psi) \overset{def}{=}  m(y_i^s(\theta);\psi). \]
  As in \citet{Gourieroux1993}, the binding function $\psi_{\infty}(.\cdot)$ is defined as:
  \[ \psi_{\infty}(\theta) \overset{def}{=} \text{argmin}_{\psi \in \Psi} \mathbb{E}\left[ \tilde m(u_i,\theta;\psi) \right]. \]
  Rather than matching sample moments, the indirect inference estimator matches functions (minimizers) of sample moment functions. 
  Assumption \ref{ass:ind} below is more detailed than the high-level conditions in \citet{Gourieroux1993}. Using implicit function arguments, it allows to express the estimator $\hat \theta_n^S$ in terms of the sample moments $\partial_\theta M_n^S$ which fit the setting of Section \ref{sec:MC_qMC} so that, eventually, Theorem \ref{th:owen} applies. 
  \begin{assumption}[Scrambled Indirect Inference] \label{ass:ind}
    Suppose that the following holds:
    \begin{itemize}
      \item[i.] The mapping $\theta \to \psi_\infty(\theta) \in \Psi$ is continuous differentiable and injective. $\Psi$ is a compact and convex subset of $\mathbb{R}^{d_\psi}$, finite-dimensional and $\psi_\infty(\theta_0) \in \text{interior}(\Psi)$.
      \item[ii.] For all $(\theta,\psi) \in \Theta \times \Psi$,
      \[ \mathbb{E}\left[ \| \tilde m(u_i,\theta;\psi) \|^2 \right] < +\infty, \]
      and there exists $C_1(\cdot,\cdot)$ such that for all $\theta \in \Theta$ and $\psi_1,\psi_2 \in \Psi$:
      \[ \| \tilde m(u_i,\theta;\psi_1)-\tilde m(u_i,\theta;\psi_2) \| \leq C_1(u_i,\theta) \times \| \psi_1-\psi_2 \|, \]
      with $\mathbb{E}[C_1(u_i,\theta)^2]<+\infty$ for all $\theta \in \Theta$.
      \item[iii.] $\tilde m$ is twice continuously differentiable in $(\theta,\psi)$, $u_i$ almost surely. For all $(\theta,\psi) \in \Theta \times \Psi$,
      \begin{align*}
        &\mathbb{E}\left[ \| \partial_{\psi} \tilde m(u_i,\theta;\psi) \|^2 \right] < +\infty, \quad 
        \mathbb{E}\left[ \| \partial^2_{\psi,\psi^\prime} \tilde m(u_i,\theta;\psi) \|^2 \right] < +\infty, \quad \mathbb{E}\left[ \| \partial^2_{\psi,\theta^\prime} \tilde m(u_i,\theta;\psi) \|^2 \right] < +\infty,
      \end{align*}
      and there exists $C_2(\cdot),C_3(\cdot),C_4(\cdot)$ such that for all $\theta_1,\theta_2 \in \Theta$ and $\psi_1,\psi_2 \in \Psi$:
      \begin{align*}
        \| \partial_{\psi}\tilde m(u_i,\theta_1;\psi_1)-\partial_{\psi}\tilde m(u_i,\theta_2;\psi_2) \| &\leq C_2(u_i) \times \left( \| \theta_1-\theta_2 \| + \| \psi_1-\psi_2 \| \right),\\
        \| \partial^2_{\psi,\psi^\prime}\tilde m(u_i,\theta_1;\psi_1)-\partial^2_{\psi,\psi^\prime}\tilde m(u_i,\theta_2;\psi_2) \| &\leq C_3(u_i) \times \left( \| \theta_1-\theta_2 \| + \| \psi_1-\psi_2 \| \right),\\
        \| \partial^2_{\psi,\theta^\prime}\tilde m(u_i,\theta_1;\psi_1)-\partial^2_{\psi,\theta^\prime}\tilde m(u_i,\theta_2;\psi_2) \| &\leq C_4(u_i) \times \left( \| \theta_1-\theta_2 \| + \| \psi_1-\psi_2 \| \right),
      \end{align*}
      with $\mathbb{E}[C_2(u_i)^2],$ $\mathbb{E}[C_3(u_i)^2]$ and $\mathbb{E}[C_4(u_i)^2]<+\infty$.   
      \item[iv.] The Hessian $\partial^2_{\psi,\psi^\prime} \mathbb{E} [ \tilde m(u_i,\theta;\psi) ]$ is  positive definite for all $\theta \in \Theta$ and all $\psi \in \Psi$ with \[0< \inf_{(\theta,\psi) \in \Theta \times \Psi} \lambda_{\min}(\partial^2_{\psi,\psi^\prime} \mathbb{E} [ \tilde m(u_i,\theta;\psi) ]) \leq \sup_{(\theta,\psi) \in \Theta \times \Psi} \lambda_{\max}(\partial^2_{\psi,\psi^\prime} \mathbb{E} [ \tilde m(u_i,\theta;\psi) ]) <+\infty.\]
      Also, $\sup_{(\theta,\psi) \in \Theta \times \Psi} \|\partial^2_{\psi,\theta^\prime} \mathbb{E} [ \tilde m(u_i,\theta;\psi) ])\| <+\infty$.
    \end{itemize} 
  \end{assumption}    

  \begin{proposition}[Consistency and Asymptotic Normality with Auxiliary Parameters] \label{prop:ind}
    Suppose Assumption \ref{ass:regular} and \ref{ass:ind} hold,  then $\hat \theta_n^S \overset{p}{\to} \theta_0$ and
    \[ \sqrt{n} \left( \hat \theta_n^S - \theta_0 \right) \overset{d}{\to} \mathcal{N}(0,\Sigma), \]
    where
    \[ \Sigma = \left( G^\prime W G \right)^{-1} G^\prime W V W G \left( G^\prime W G \right)^{-1},\]
    $G = \partial_\theta \mathbb{E} [\hat \psi_n^S(\theta_0)]$, $V  =\lim_{n \to \infty} n \times \text{var}(\hat \psi_n)$.
  \end{proposition}
  Proposition \ref{prop:ind} is similar to the results found in \citet{Gourieroux1993} but here the simulation noise has no effect on the asymptotic variance as in Proposition \ref{prop:static_nocovariates}.

  \subsection{Dynamic Models} \label{sec:Dynamic}
  \subsubsection{qMC-only Estimator}
  For simplicity, write:
  \[ \hat \psi_T^S(\theta) = \frac{1}{T S} \sum_{t=1}^{T S} \tilde \psi(u_t;\theta), \]
  where $u_t$ has the appropriate dimension $d$ given in Section \ref{sec:SMM_dyn_qMC}. For the qMC-only estimator, $\hat \psi_T^S$ is simply a cross-sectional average over short-time series. This fits the framework of Section \ref{sec:MC_qMC} directly and under the conditions in Assumption \ref{algo:static_no_covariates} the estimator is consistent and asymptotically normal as shown in the Proposition below. As for the static case, the asymptotic variance is not inflated by the simulation noise, even for $S=1$.
  \begin{proposition}[Consistency and Asymptotic Normality - qMC only] \label{prop:dyn_case1}
    Suppose Assumptions \ref{ass:regular} and \ref{ass:smooth_nocovariates} hold and the draws are generate as in Algorithm \ref{algo:ScrMMdyn1} then $\hat \theta_n^S \overset{p}{\to} \theta_0$ and:
    \[ \sqrt{T}\left( \hat \theta_T^S -\theta_0 \right) \overset{d}{\to} \mathcal{N}(0,\Sigma), \]
    where
    \[ \Sigma = \left( G^\prime W G \right)^{-1} G^\prime W V W G \left( G^\prime W G \right)^{-1}, \]
    $G = \partial_\theta \mathbb{E} \left[ \hat \psi_T^S(\theta_0)\right]$, $V = \lim_{T \to \infty} T \times \text{var}(\hat \psi_T)$.
  \end{proposition}

  \subsubsection{MC-qMC Hybrid Estimator}
  For simplicity, write:
  \[ \hat \psi_T^S(\theta) = \frac{1}{T S} \sum_{t=1}^{T S} \tilde \psi(y^1_t,z^1_t,u_t;\theta), \]
  where $y_t^1,z_t^1$ are simulated using MC methods as in Algorithm \ref{algo:ScrMMdyn2}. The hybrid MC-qMC approach relies on MC simulations to approximately draw initial values from the ergodic distribution and is combined with the scramble to simulate a cross-section of paths.
  \begin{assumption}[Dynamic Models - MC-qMC] \label{ass:dyn_case2}
    Suppose there exists a constant $K >0$ such that:
    \begin{itemize}
      \item[i.] For all $\theta \in \Theta$, $( y_t^1, z_t^1)$ is geometrically ergodic: $\|f_t( y_t^1, z_t^1;\theta)-f_\infty( y_t^\infty, z^\infty_t;\theta)\|_{TV} \leq C_1 \times \rho^t$, for some $\rho \in [0,1)$ and $0\leq C_1 < +\infty$, where $f_\infty$ is the ergodic distribution of $y_t^1,z_t^1$ and $f_t$ its non-stationary distribution with fixed starting value.
      \item[ii.] For all $\theta \in \Theta$, $\mathbb{E}[\|\tilde \psi( y_t^1, z_t^1,u_t;\theta)\|^{4}|u_t] \leq K < + \infty$, $\mathbb{E}[\|\tilde \psi( y^\infty_t,  z^\infty_t,u_t;\theta)\|^{4}|u_t] \leq K < + \infty$.
      \item[iii.]  For any $\|\theta_1-\theta_2\|$ small, $\|\tilde \psi( y_t^1, z_t^1,u_t;\theta_1) - \tilde \psi( y_t^1, z_t^1,u_t;\theta_2)\| \leq C_2( y_t^1, z_t^1,u_t;\theta_1) \times \|\theta_1-\theta_2\|$ with $\mathbb{E}[\| C_2( y_t^1, z_t^1,u_t;\theta_1) \|^{4}|u_t]\leq K < +\infty$ and $\mathbb{E}[\| C_2( y^\infty_t, z^\infty_t,u_t;\theta_1) \|^{4}|u_t]\leq K < +\infty$.
      \item[iv.] For all $\theta \in \Theta$, $\mathbb{E}[\| \partial_\theta \tilde \psi( y_t^1,  z_t^1,u_t;\theta)\|^{4}|u_t] \leq K < + \infty$, $\mathbb{E}[\| \partial_\theta \tilde \psi( y^\infty_t,  z^\infty_t,u_t;\theta)\|^{4}|u_t] \leq K < + \infty$.
      \item[v.]  For any $\|\theta_1-\theta_2\|$ small, $\|\partial_\theta \tilde \psi( y_t^1, z_t^1,u_t;\theta_1) - \partial_\theta \tilde \psi( y_t^1, z_t^1,u_t;\theta_2)\| \leq C_3( y_t^1, z_t^1,u_t;\theta_1) \times \|\theta_1-\theta_2\|$ with $\mathbb{E}[\| C_3( y_t^1, z_t^1,u_t;\theta_1) \|^{4}|u_t]\leq K < +\infty$  and $\mathbb{E}[\| C_3( y^\infty_t, z^\infty_t,u_t;\theta_1) \|^{4}|u_t]\leq K < +\infty$.
      \item[vi.] $\lim_{T \to \infty} T \times \text{var}( \hat \psi_T^S(\theta_0) |u_1,\dots,u_{TS})$ is positive definite and finite.
    \end{itemize}
  \end{assumption}
  Assumption \ref{ass:dyn_case2} i. is the usual geometric ergodicity condition \citep{duffie-singleton}. Conditions ii.-v. are more restrictive, they hold if the moments are bounded. To relax these conditions, one would need to extend the CLT in Theorem 5.20 of \citet{White1984} to unbounded non-identically distributed dependent arrays which is outside the scope of this paper. Condition vi. requires the variance to be non-degenerate to apply a CLT. Otherwise, simulation noise is negligible in some directions which is not problematic in this setting.

  \begin{proposition}[Consistency and Asymptotic Normality - MC-qMC] \label{prop:dyn_case2}
    Suppose Assumptions \ref{ass:regular} and \ref{ass:dyn_case2} hold, then $\hat \theta_T^S \overset{p}{\to} \theta_0$ and
    \[ \sqrt{T} (\hat \theta_T^S - \theta_0) \overset{d}{\to} \mathcal{N}(0,\Sigma), \]
    where
    \[ \Sigma = \left( G^\prime W G \right)^{-1} G^\prime W V W G \left( G^\prime W G \right)^{-1}, \]
    $G = \lim_{T \to \infty} \mathbb{E}[\partial_\theta \hat \psi_T^S(\theta_0)]$, $V = \lim_{T \to \infty} T \times \text{var}(\hat \psi_T^S(\theta_0)|u_1,\dots,u_{TS})$.
  \end{proposition}
  Proposition \ref{prop:dyn_case2} is similar to \citet{duffie-singleton}, the main idea to prove the result is to write the simulated moments as the sum of a mixing non-identically distributed heterogeneous array and an average of identically distributed non-independent terms. As in Proposition \ref{prop:static_covariates}, the former is handled using a specific CLT while the former uses Theorem \ref{th:owen} with similar steps to \citet{Jennrich1969}.

  \section{Monte-Carlo Illustrations} \label{sec:MC}
  The following illustrates the finite sample properties of the Scrambled Method of Moments and Scrambled Indirect Inference in several simple examples and one application drawn from the heterogeneous agents literature. All simulations were carried out in R and C++ using the \textit{Rcpp} package. Scrambled sequences were generated using the \textit{fOptions} package.
  \subsection{Simple Examples}  \label{sec:simple}
    
  \subsubsection{Mean-Variance} \label{sec:mean_var}
  The first example, drawn from \citet{Gourieroux1993}, considers the estimation of a sample mean and variance of for an \textit{iid} Gaussian sample:
  \[ y_i = \mu + \sigma e_i, \quad e_i \sim \mathcal{N}(0,1). \]
  This example illustrates Algorithms \ref{algo:SMM}, \ref{algo:static_no_covariates} and Proposition \ref{prop:ind}. As in the original paper, the auxiliary parameters $\hat \psi_n$ are the sample mean and variance of $(y_1,\dots,y_n)$:
  \[ \hat \psi_n = ( \hat \mu_n, \hat \sigma_n^2 )^\prime = \frac{1}{n} \sum_{i=1}^n (y_i, [y_i - \hat \mu_n]^2)^\prime. \]
  In the $5,000$ Monte-Carlo replications, the sample size is $n=100$ and $\theta_0=(\mu_0,\sigma_0^2) = (0,1)$. The number of simulated samples is $S=1,2,4$ and $20$.   For SMM, $e_i^s$ is drawn using the random number generator \textit{rnorm} in $R$ and antithetic draws are generated for $S=2,4$ and $20$ by taking $e_i^{s+S/2} = -e_i^{s}$ for each $s=1,\dots,S/2$. The \textit{fOptions} package generates the scrambled  Gaussian shocks directly.   Table \ref{tab:mean_variance} summarizes the biases and standard deviations of the estimators. 
  \begin{table}[H] \caption{Mean and Variance Estimation} \label{tab:mean_variance}
    \begin{center}
      { \footnotesize
      \begin{tabular}{l|c|cccc|cccc|cccc}
        \hline \hline
                 & MM & \multicolumn{4}{|c|}{SMM} & \multicolumn{4}{c|}{Antithetic} & \multicolumn{4}{c}{Scramble}\\
                coef./$S$ &   & 1 & 2 & 4 & 20 & 1 & 2 & 4 & 20 & 1 & 2 & 4 & 20\\ \hline
                $\sqrt{n} \times \text{std}(\hat \mu)$ & 0.99 & 1.44 & 1.22 & 1.10 & 1.01 & - & 0.99 & 1.00 & 0.98 & 1.00 & 1.00 & 1.00 & 1.00 \\
                $\sqrt{n} \times \text{std}(\hat \sigma^2)$ & 1.41 & 2.07 & 1.76 & 1.60 & 1.47 & - & 2.03 & 1.75 & 1.50 & 1.44 & 1.44 & 1.41 & 1.41 \\
                $100 \times \text{bias}(\hat \sigma^2)$ & -0.93 & 2.38 & 0.89 & 0.49 & 0.25 & - & 1.93 & 0.98 & 0.44 & -0.43 & -0.87 & -1.07 & -0.80 \\
        \hline \hline
      \end{tabular} 
    }
    \end{center}
  \end{table}
  Because it has no simulation noise, the Method of Moments (MM) estimator has the smallest variance. SMM has a bias correction property for $\hat \sigma^2$ \citep{Gourieroux1993}. For $\hat \mu_n$, antithetic draws and the scramble perform equally well for $S=2$ and $S=1$, respectively. For $\hat \sigma_n^2$, antithetic draws perform worse than SMM and the scramble. This is in line with the discussion in Section \ref{sec:MC_anti}. The scramble performs similarly to the MM while SMM requires $S=20$ to perform similarly. SMM and antithetic draws reduce the bias while the scramble does not. This reflects the fact that the scrambled $\hat \psi_n^S(\theta_0)$ approximates the asymptotic binding function $\psi_\infty(\theta_0) = \lim_{n \to \infty} \mathbb{E}(\hat \psi_n)$ while SMM and antithetic draws approximate the binding function $\psi(\theta_0) = \mathbb{E}(\hat \psi_n)$ which provides some finite sample bias correction.

  \subsubsection{Probit Model} \label{sec:probit}
  The second example illustrates Algorithm \ref{algo:static_covariates} with non-smooth moments and covariates. The DGP is a Probit model:
  \[ y_i = \mathbbm{1} \left\{ \theta_1 + \theta_2 x_i \theta + e_i \geq 0 \right\}, \quad e_i \overset{iid}{\sim} \mathcal{N}(0,1), \quad x_i \sim \mathcal{N}(0,1). \]
  The moments $\hat \psi_n$ consist of the intercept and the slope in an OLS regression of $y_i$ on $x_i$. In the $5,000$ Monte-Carlo replications, the sample size is $n=1,000$ and $\theta_0=(\theta_{1,0},\theta_{2,0}) = (1,1)$. The number of simulated samples is $S=1,2,4$ and $20$. The standard deviations of the estimators are reported in Table \ref{tab:probit}.
  \begin{table}[H] \caption{Probit Models: $\sqrt{n} \times \text{std}(\hat \theta_n^S)$} \label{tab:probit}
    \centering
    \begin{tabular}{r|cccc|cccc|cccc}
      \hline \hline
      & \multicolumn{4}{c|}{SMM} & \multicolumn{4}{c|}{Antithetic} & \multicolumn{4}{c}{Scramble}\\
      coef./$S$     & $1$ & $2$ & $4$ & $10$ & $1$ & $2$ & $4$ & $10$ & $1$ & $2$ & $4$ & $10$ \\ 
      \hline
      $\hat \theta_{1,n}$    & 2.38 & 2.24 & 2.11 & 1.91 & - & 2.17 & 2.06 & 1.91 & 2.14 & 2.09 & 2.01 & 1.90\\ 
      $\hat \theta_{2,n}$    & 2.76 & 2.57 & 2.42 & 2.22 & - & 2.47 & 2.35 & 2.21 & 2.68 & 2.52 & 2.39 & 2.19\\
       \hline \hline
    \end{tabular}
    \end{table} 
  The Scrambled Method of Moments outperforms SMM for $S=1$ and above. For $S \geq 2$, the scramble performs similarly to antithetic draws for estimating $\theta_0$ and $\theta_1$. The gains are less substantial than in the previous example.

  \subsubsection{ARMA Model} \label{sec:arma}
  To illustrate Algorithms \ref{algo:ScrMMdyn1} and \ref{algo:ScrMMdyn2} consider the following ARMA(1,1) model:\footnote{In the notation of Section \ref{sec:Dynamic}, the model can be written as: $y_t = \rho y_{t-1} + \sigma [z_{t,1} + \vartheta z_{t,2}],$ $(z_{t,1},z_{t,2})^\prime = (e_{t},z_{t-1,1})^\prime,$
  with $z_{t} = (z_{t,1},z_{t,2})^\prime$.}
  \[ y_t = \rho y_{t-1} + \sigma[ e_t + \vartheta e_{t-1}], \quad e_t \overset{iid}{\sim} \mathcal{N}(0,1), \]
  In the $5,000$ Monte-Carlo replications, the sample size is $T=200$ and $\theta_0=(\vartheta_0,\rho_0,\sigma_0^2) = (0.5,0.5,1)$. The number of simulated samples is $S=1,2$. The moments are the OLS coefficients from regressing $y_t$ on its first $L=4$ lags and the variance of the OLS residuals. Using auto-covariances as moments instead yields similar results.  
  
  Algorithm \ref{algo:ScrMMdyn1} requires sampling $(y^1_t,e^1_t)$ from its stationary distribution directly. The marginals are known since $e^1_t \sim \mathcal{N}(0,1)$ by assumption and $y^1_t \sim \mathcal{N}(0, [1+\vartheta^2 + 2\rho \vartheta]/[1-\rho^2]\sigma^2)$. Since they are jointly Gaussian, it is sufficient to compute their covariance, $\text{cov}(e_t^1,y_t^1) = \rho \vartheta \sigma$, to find their joint distribution: 
  \[ \left( \begin{array}{c} y_t^1 \\ e_t^1 \end{array} \right) \sim \mathcal{N} \left( \left( \begin{array}{c} 0 \\ 0 \end{array} \right), \left( \begin{array}{cc} \frac{1+\vartheta^2 + 2\rho \vartheta}{1-\rho^2}\sigma^2 & \rho \vartheta \sigma \\ \rho \vartheta \sigma & 1 \end{array} \right)\right). \]
  Transforming independent bivariate scrambled Gaussian shocks into draws from the joint distribution above is then straightforward.
  For Algorithm \ref{algo:ScrMMdyn2}, the $(y^1_t,e^1_t)$ need to be sampled using MC methods. First, the initial value $(y^1_0,e^1_0)=(0,0)$ is set and a path $(y_t^1,e_t^1)$ is simulated with random MC draws.  
  Once these $(y_t^1,e_t^1)$ are simulated, the remaining $(y_t^2,e_t^2,\dots,y_t^5,e_t^5)$ are computed using scrambled Gaussian shocks.\footnote{Here the number of lags used to compute the moments is $L=5$ because $y_t^5$ is regressed on its 4 lags $y_t^4,\dots,y_t^1$.}

  Table \ref{tab:arma} compares the Maximum Likelihood Estimator (MLE) with SMM and antithetic draws, the qMC-only scramble from Algorithm \ref{algo:ScrMMdyn1} (reported in the Scramble column) and the hybrid MC-qMC scrambled from Algorithm \ref{algo:ScrMMdyn2} (reported in the Scramble-MC column).
  \begin{table}[H] \caption{ARMA(1,1): $\sqrt{n} \times \text{std}(\hat \theta_n^S)$} \label{tab:arma}
    \centering
    \begin{tabular}{r|c|cc|cc|cc|cc}
      \hline \hline
                      & MLE  & \multicolumn{2}{c|}{SMM} & \multicolumn{2}{c|}{Antithetic} & \multicolumn{2}{c|}{Scramble} & \multicolumn{2}{c}{Scramble-MC}\\
      coef./$S$       &      & 1    & 2    & 1 & 2     & 1    & 2     & 1     & 2\\ \hline
      $\hat \rho_n$   & 1.10 & 1.64 & 1.44 & - & 1.66  & 1.20 & 1.17  & 1.39  & 1.28\\ 
      $\hat \theta_n$ & 1.13 & 1.86 & 1.57 & - & 1.87  & 1.33 & 1.28  & 1.53  & 1.36\\ 
      $\hat \sigma_n$ & 0.72 & 1.05 & 0.90 & - & 1.04  & 0.76 & 0.72  & 0.94  & 0.84\\
      \hline \hline
    \end{tabular}
  \end{table}
  MLE corresponds to the lower bound for variance of the estimators. The qMC-only scramble from Algorithm \ref{algo:ScrMMdyn1} outperforms SMM and antithetic draws. Antithetic draws perform worse than SMM using the same $S$ which further illustrates the discussion in Section \ref{sec:MC_anti}. The hybrid MC-qMC Algorithm \ref{algo:ScrMMdyn2} performs better than SMM and antithetic draws and, as expected, worse than the qMC-only approach. 
  

  \subsection{An Income Process with ``Lots of Heterogeneity''} \label{sec:het_income}
  The last example is a more substantial model borrowed from \citet{BROWNING2010}.\footnote{The data generating process considered here involves all the coefficients found in \citet{BROWNING2010}, Table 2 minus the measurement errors and the time-trend in the ARCH component which are not considered in this Monte-Carlo exercise.} Simulation-based estimation is commonly used in this heterogeneous agents literature due to the complexity and intractability of the models.\footnote{See e.g. \citet{Guvenen2011} for an overview of the computation and estimation of heterogeneous agents models.} The baseline data generating process is an ARMA(1,1) at the individual level:
  \begin{align*}
    y_{i,t} = \delta_i \times \left( [1-\omega^t_{i}] + \beta_i  \times [1-\omega^{t-1}_{i}]\right) + \alpha_i\beta_i  +\beta_i y_{i,t-1} + \alpha_i \times [1-\beta_i] \times t + \varepsilon_{i,t} + \theta_i \varepsilon_{i,t-1}
  \end{align*}
  where the drift $\alpha_i$, long-run mean $\delta_i$, AR and MA coefficients $\beta_i,\theta_i$ as well as the persistence coefficient $\omega_i$ all vary at the individual level. The Gaussian shocks to log-income in the time dimension are denoted by $\varepsilon_i$ while the Gaussian shocks to the ARMA coefficients $\alpha_i,\beta_i,\dots$ will be denoted by $\eta_i$. The initial value for log-income $y_{i,0}$ is drawn as:
  \begin{align*}
    y_{i,0} = \exp( \tau )\times \eta_{i,0}.
  \end{align*}
  The heterogenous ARMA coefficients are then drawn using:
  \begin{align*}
      \nu_{i,0} &= \exp(\phi_{11} + \phi_{12}\times y_{i,0} + \psi_{11} \times \eta_{i,1})\\
      \theta_{i} &= \text{logit}\left( \phi_{21} + \phi_{22} \times y_{i,0} + \psi_{21} \times \eta_{i,1} + \psi_{22} \times \eta_{i,2} \right) -1/2\\
      \alpha_i &= \phi_{31} + \phi_{32} \times y_{i,0} + \psi_{3,1} \times \eta_{i,1}\\
      \beta_{i} &= \text{logit}\left( \phi_{41} + \phi_{42} \times y_{i,0} + \psi_{4,1} \times \eta_{i,1}+ \psi_{4,2} \times \eta_{i,2} \right)\\
      \delta_{i} &= \phi_{51} + \phi_{52} \times y_{i,0} + \psi_{5,1} \times \eta_{i,1}+ \psi_{5,2} \times \eta_{i,2}\\
      \omega_{i} &= \text{logit}( \phi_{61} + \psi_{62} \times \eta_{i,2})
  \end{align*}
  where $\text{logit}$ is the usual logistic transformation $\text{logit}(x) = 1/[1+\exp(-x)]$. $\eta_{i,0},\dots,\eta_{i,2} \overset{iid}{\sim} \mathcal{N}(0,1)$. For a discussion of the parameters and the role of the transformations, see \citet{BROWNING2010}. $\nu_{i,0}$ is the initial value for the ARCH-type heteroskedasticity in the shocks $\varepsilon_{i,t}$:
  \begin{align*}
    \sigma_{i,1}^2 &= \nu_{i,0}, &\varepsilon_{i,1} = \sigma_{i,1} \times e_{i,1}\\
    \sigma_{i,t}^2 &= \nu_{i,0} + \text{logit}(\varphi) \times \varepsilon_{i,t-1}^2, &\varepsilon_{i,t} = \sigma_{i,t} \times e_{i,t}
  \end{align*}
  where $e_{i,0},\dots,e_{i,T} \overset{iid}{\sim} \mathcal{N}(0,1)$.  In the simulations, the number of households is $n=1,000$; the number of time periods is $T=30$. As in the original paper, a burn-in period of $T_{burn}=3$ periods is used to reduce the effect of the initial conditions. The parameter values are taken from Table 2 in \citet{BROWNING2010} and the moments are those described in their Appendix A.2 except the ones involving year of birth which are not considered in these simulations. In a nutshell, the moments involve the aggregation of individual-level OLS coefficients, moments based on OLS residuals, autocorrelations and measures of social mobility.  

  \begin{table} \caption{Income Process with Heterogeneity: $\sqrt{n} \times \text{std}(\hat \theta_n^S)$} \label{tab:MC_het} 
    \centering
    \begin{tabular}{r|ccc|ccc|ccc|ccc} 
      \hline \hline
       & \multicolumn{3}{c|}{SMM} & \multicolumn{3}{c|}{Antithetic} & \multicolumn{6}{c}{Scramble}   \\   \cline{8-13}
       & \multicolumn{3}{c|}{ }   & \multicolumn{3}{c|}{ }          & \multicolumn{3}{c|}{$S$ samples of size $n$} &  \multicolumn{3}{c}{$1$ sample of size $n S$}\\
     coef./$S$ & $1$ & $2$ & $4$ & $1$ & $2$ & $4$ & $1$ & $2$ & $4$ & $1$ & $2$ & $4$ \\ 
      \hline
$\tau$      & 1.27 & 1.20 & 1.20 & \,\,\,-\,\, & 1.25 & 1.35 & 1.06 & 1.12 & 1.11 & 1.06 & 1.13 & 1.13 \\
$\phi_{11}$ & 1.29 & 1.22 & 1.18 & \,\,\,-\,\, & 1.06 & 1.23 & 1.00 & 0.99 & 1.00 & 1.00 & 1.03 & 1.04 \\
$\phi_{12}$ & 4.32 & 3.70 & 3.28 & \,\,\,-\,\, & 3.70 & 4.47 & 3.39 & 3.07 & 2.95 & 3.39 & 3.09 & 3.20\\
$\phi_{21}$ & 1.62 & 1.44 & 1.39 & \,\,\,-\,\, & 1.53 & 1.73 & 1.30 & 1.31 & 1.26 & 1.30 & 1.20 & 1.28\\
$\phi_{31}$ & 0.18 & 0.15 & 0.14 & \,\,\,-\,\, & 0.14 & 0.15 & 0.15 & 0.14 & 0.13 & 0.15 & 0.14 & 0.14\\
$\phi_{32}$ & 0.17 & 0.16 & 0.14 & \,\,\,-\,\, & 0.15 & 0.15 & 0.15 & 0.15 & 0.14 & 0.15 & 0.14 & 0.14\\
$\phi_{41}$ & 1.90 & 1.66 & 1.68 & \,\,\,-\,\, & 1.72 & 1.84 & 1.60 & 1.57 & 1.54 & 1.60 & 1.60 & 1.59\\
$\phi_{51}$ & 2.98 & 2.84 & 2.70 & \,\,\,-\,\, & 2.68 & 2.79 & 2.94 & 2.78 & 2.59 & 2.94 & 2.76 & 2.61\\
$\phi_{52}$ & 8.54 & 8.17 & 7.58 & \,\,\,-\,\, & 7.75 & 7.73 & 8.00 & 7.59 & 7.23 & 8.00 & 7.54 & 7.31\\
$\phi_{61}$ & 4.18 & 4.20 & 3.29 & \,\,\,-\,\, & 3.36 & 3.32 & 3.27 & 3.16 & 2.87 & 3.27 & 3.01 & 3.04\\
$\psi_{11}$ & 1.25 & 1.29 & 1.20 & \,\,\,-\,\, & 1.29 & 1.28 & 0.89 & 0.88 & 0.96 & 0.98 & 0.97 & 1.01\\
$\psi_{22}$ & 2.55 & 2.42 & 2.39 & \,\,\,-\,\, & 2.47 & 2.89 & 2.13 & 2.08 & 2.15 & 2.13 & 2.09 & 2.21\\
$\psi_{31}$ & 0.08 & 0.07 & 0.07 & \,\,\,-\,\, & 0.07 & 0.07 & 0.06 & 0.07 & 0.06 & 0.06 & 0.06 & 0.06\\
$\psi_{41}$ & 2.68 & 2.36 & 2.30 & \,\,\,-\,\, & 2.40 & 2.61 & 2.27 & 2.22 & 2.10 & 2.27 & 2.22 & 2.19\\
$\psi_{42}$ & 1.90 & 1.76 & 1.74 & \,\,\,-\,\, & 1.75 & 1.87 & 1.56 & 1.53 & 1.55 & 1.56 & 1.63 & 1.58\\
$\psi_{51}$ & 3.30 & 3.13 & 2.83 & \,\,\,-\,\, & 2.90 & 3.28 & 2.64 & 2.89 & 2.56 & 2.64 & 2.64 & 2.70\\
$\psi_{52}$ & 2.14 & 2.03 & 2.20 & \,\,\,-\,\, & 2.14 & 2.45 & 1.78 & 1.95 & 1.83 & 1.78 & 1.84 & 1.96\\
$\psi_{62}$ & 3.53 & 3.53 & 3.65 & \,\,\,-\,\, & 3.55 & 4.01 & 3.34 & 3.28 & 3.36 & 3.34 & 3.44 & 3.56\\
$\varphi$   & 2.24 & 2.10 & 2.39 & \,\,\,-\,\, & 2.76 & 3.06 & 1.99 & 1.94 & 2.21 & 1.99 & 2.30 & 2.56\\
       \hline \hline
    \end{tabular}
    \end{table}
  
  The implementation of SMM is standard and described in Appendix A.4 of the original paper. For the scramble, a $(n\times S) \times (T+T_{burn} + 3) = (1,000 \times S) \times 36$ matrix of scrambled standard gaussian shocks is drawn. The integration dimension $d=36$ is sufficiently large to illustrate the finite sample performance of the scrambled method of moments with a relatively large number of shocks. The first three dimensions (columns of the matrix) correspond to $\eta_{i,0},\dots,\eta_{i,2}$, the remaining dimensions correspond to time dimensions $e_{i,1},\dots,e_{i,T+T_{burn}}$. The rows correspond to the cross-sectional dimension of the shocks, \textit{i.e.} the $i=1,\dots,n\times S$ index.

  The results from the $2,000$ Monte-Carlo replications are presented in Table \ref{tab:MC_het} for $S=1,2,4$. SMM and antithetic draws are used as a benchmark for the scramble with either a large sample of $n \times S$ individuals (as in Algorithm \ref{algo:static_no_covariates} or $S$ samples of $n$ individuals (as in Algorithm \ref{algo:static_no_covariates}). The scramble generally outperform SMM and antithetic draws. Antithetic draws either under or over-performs SMM depending on the parameter of interest which is in line with previous discussions. Both implementations of the Scrambled Method of Moments perform similarly. For some coefficients, there is little to no improvement in increasing $S$ from $1$ to $2$ or $4$. For most coefficients, the scramble with $S=2$ outperforms SMM with $S=4$. 
  Furthermore, using the same $S=4$, the replications were computed about $15\%$ faster for the scramble than SMM. Since the only difference between the two is the shocks used in the simulations, this reflects faster convergence of the optimizer. Possibly because the scramble are smoother (less noisy) than the MC moments which makes the objective function easier to minimize.
  \section{Conclusion} \label{sec:Conclusion}
    This paper proposes several algorithms implementing Owen's scramble for simulation-based estimation. Since the method is designed for computing integrals of \textit{iid} sequences, some care is needed when simulating data with covariates or time series. Large sample results are provided to support the proposed algorithms. The results for dynamic models could be extended to non-smooth bounded moments through additional stochastic equicontinuity results using the inequality in \citet{Andrews1994} for instance. The simulations illustrate the finite performance of the Scrambled Methods of Moments and Scrambled Indirect Inference compared to other commonly used methods. The last example suggests the scramble could be useful in larger scale problems found in the heterogenous agents literature where SMM is commonly used.
\newpage
\bibliography{refs}
\newpage
\begin{appendices}
  \renewcommand\thetable{\thesection\arabic{table}}
  \renewcommand\thefigure{\thesection\arabic{figure}}
  \renewcommand{\theequation}{\thesection.\arabic{equation}}
  \renewcommand\thelemma{\thesection\arabic{lemma}}
  \renewcommand\thetheorem{\thesection\arabic{theorem}}
  \renewcommand\thedefinition{\thesection\arabic{definition}}
    \renewcommand\theassumption{\thesection\arabic{assumption}}
  \renewcommand\theproposition{\thesection\arabic{proposition}}
    \renewcommand\theremark{\thesection\arabic{remark}}
    \renewcommand\thecorollary{\thesection\arabic{corollary}}

    \section{Proofs} \label{sec:proofs}

    \subsection{Static Models}
    \subsubsection{Smooth Moments with No Covariates}

    \begin{lemma}[ULLN and CLT for Smooth Moments without Covariates] \label{lem:ULLNnocovariates}
      Suppose the conditions in Assumption \ref{ass:smooth_nocovariates} hold, then:
      \begin{itemize}
      \item[i.] $\sup_{\theta \in \Theta} \|\hat \psi_n^s(\theta) - \mathbb{E}(\hat \psi_n^s(\theta))\| = o_p(1),$
      \item[ii.] $\sup_{\theta \in \Theta} \| \partial_\theta \hat \psi_n^s(\theta) - \partial_\theta \mathbb{E}(\hat \psi_n^s(\theta))\| = o_p(1),$
      \item[iii.] $\|\hat \psi_n^s(\theta_0) - \mathbb{E}[\hat \psi_n^s(\theta_0)]\| = o_p(n^{-1/2}).$
      \end{itemize} 
    \end{lemma}
  
    \begin{proof}[Proof of Lemma \ref{lem:ULLNnocovariates}] $\,$ \newline
      \noindent \textbf{Part i.} ULLN for $\hat \psi_n^S(\theta)$\\
      Assumption \ref{ass:smooth_nocovariates} implies $\|\hat \psi_n^S(\theta) - \mathbb{E}[\hat \psi_n^S(\theta)]\| = o_p(1)$ pointwise. Using the same steps as in \citet{Jennrich1969}, the Lipschitz condition implies that for a finite cover $(\theta_1,\dots,\theta_J)$ of $\Theta$:
      \begin{align*}
        \sup_{\theta \in \Theta} &\|\hat \psi_n^S(\theta)-\mathbb{E}(\hat \psi_n^S(\theta))\| \\&\leq \max_{ j  \in \{1,\dots,J\}} \|\hat \psi_n^S(\theta_j)-\mathbb{E}(\hat \psi_n^S(\theta_j))\| +   \sup_{\theta \in \Theta} \min_{ j  \in \{1,\dots,J\}} \Big\|[\hat \psi_n^S(\theta)-\hat \psi_n^S(\theta_j)]-[\mathbb{E}(\hat \psi_n^S(\theta))-\mathbb{E}(\hat \psi_n^S(\theta_j))]\Big\|.
      \end{align*}
      Using the Lipschitz condition for $\tilde \psi$, we have:
      \begin{align*}
        \sup_{\theta \in \Theta} \min_{j \in \{1,\dots J\} }  &\Big\|[\hat \psi_n^S(\theta)-\hat \psi_n^S(\theta_j)]-[\mathbb{E}(\hat \psi_n^S(\theta))-\mathbb{E}(\hat \psi_n^S(\theta_j))] \Big\| \\ &\leq \Big[\frac{1}{n} \sum_{i=1}^n C_1(u_i) + \mathbb{E}(C_1(u_i)) \Big] \times \sup_{\theta \in \Theta} \min_{j \in \{1,\dots,J\}} \|\theta-\theta_j\| 
      \end{align*}
      Since $C_1$ is square integrable, Theorem \ref{th:owen} applies to $C_1(u_i)$ so that $\Big[\frac{1}{n} \sum_{i=1}^n C_1(u_i) + \mathbb{E}(C_1(u_i)) \Big] = 2 \times \mathbb{E}[C_1(u_i)] + o_p(1)$. For $J \geq 1$ large enough and an appropriate cover, $\sup_{\theta \in \Theta} \min_{j \in \{1,\dots,J\}} \|\theta-\theta_j\| \leq \frac{\varepsilon}{4 \mathbb{E}[C_1(u_i)]}$. Similarly, for any given $J \geq 1$ fixed, $\max_{ j  \in \{1,\dots,J\}} \|\hat \psi_n^S(\theta_j)-\mathbb{E}(\hat \psi_n^S(\theta_j))\| \leq \varepsilon/2$ with probability going to $1$ as $n \to \infty$. Overall, this implies that:
      \[ \mathbb{P}(\sup_{\theta \in \Theta} \|\hat \psi_n^S(\theta)-\mathbb{E}(\hat \psi_n^S(\theta))\|>\varepsilon) \to 0, \]
      this provides a ULLN with scrambled draws. \newline

      \noindent \textbf{Part ii.} ULLN for $\hat \psi_n^S(\theta)$\\
      The ULLN can be directly applied to $\partial_\theta \hat \psi_n^s(\theta)$ under the stated assumptions. \newline

      \noindent \textbf{Part iii.} Convergence rate for $\hat \psi_n^S(\theta_0)-\mathbb{E}[\hat \psi_n^S(\theta_0)]$\\
      This is a direct application of Theorem \ref{th:owen} which concludes the proof.
    \end{proof}

    \begin{proof}[Proof of Proposition \ref{prop:static_nocovariates}]
      Combining Assumption \ref{ass:regular} with the ULLN in Lemma \ref{lem:ULLNnocovariates} imply that the consitency Theorem 2.1 in \citet{Newey1994a} applies; \textit{i.e.} $\hat \theta_n^S \overset{p}{\to} \theta_0$. Then, the ULLN for the Jacobian with a mean value expansion argument imply:
      \begin{align*}
        \sqrt{n} \left( \hat \theta_n^S -\theta_0 \right) &= - \left(  G^\prime W G \right)^{-1} G^\prime W \sqrt{n}\left[ \underbrace{\hat \psi_n - \mathbb{E}[\hat \psi_n^S(\theta_0)]}_{= O_p(n^{-1/2})} + \underbrace{\mathbb{E}[\hat \psi_n^S(\theta_0)] - \hat \psi_n^S(\theta_0)}_{=o_p(n^{-1/2})} \right] +o_p(1)\\
        &= - \left(  G^\prime W G \right)^{-1} G^\prime W \sqrt{n}\left[ \hat \psi_n - \mathbb{E}[\hat \psi_n^S(\theta_0)] \right] + o_p(1)\\
        &\overset{d}{\to} \mathcal{N}(0,\Sigma),
      \end{align*}
      where $\Sigma$ is defined in the Proposition. This concludes the proof.
    \end{proof}

    \subsubsection{Non-Smooth Moments with Covariates}

    \begin{lemma}[Stochastic Equicontinuity and CLT with Covariates] \label{lem:se_clt_covariates}
      Suppose that Assumptions \ref{ass:regular} and \ref{ass:smooth} hold and $S=1$, then:
      \begin{itemize} 
        \item[i.] $\sqrt{n}\left[ \hat \psi_n - \hat \psi_n^S(\theta_0) \right] \overset{d}{\to} \mathcal{N}(0,\tilde V)$ where $\tilde V = \mathbb{E} \left[ \text{var} \left( \psi(y_i,x_i) - \tilde \psi(x_i,u_i;\theta) | u_i \right) \right]$
        \item[ii.] $\sup_{\|\theta_1-\theta_2\| \leq \delta_n} \sqrt{n}\| [\hat \psi_n^S(\theta_1)-\hat \psi_n^S(\theta_2)] - \mathbb{E}[\hat \psi_n^S(\theta_1)-\hat \psi_n^S(\theta_2)|u_{1},\dots,u_n] \| = o_p(1)$, $\forall \delta_n \searrow 0$
        \item[iii.]   $\sup_{\|\theta_1-\theta_2\| \leq \delta_n} \| \mathbb{E}[\hat \psi_n^S(\theta_1)-\hat \psi_n^S(\theta_2)|u_{1},\dots,u_n] - \partial_\theta \mathbb{E}[\hat \psi_n^S(\theta_2)](\theta_1-\theta_2) \| \leq O_p(\delta_n^2)$, $\forall \delta_n \searrow 0$
      \end{itemize}
    \end{lemma} 
  
    \begin{proof}[Proof of Lemma \ref{lem:se_clt_covariates}] $\,$ \newline
      \noindent \textbf{Part i.} CLT for $\hat \psi_n - \hat \psi_n^S(\theta_0)$\\
      Similarly to \citet{Okten2006}, the main idea is to verify the conditions for an independent non-identically distributed CLT hold holding the qMC draws $u_1,\dots,u_{n}$ fixed. Note that:
      \[ \hat \psi_n - \hat \psi_n^S(\theta_0) = \underbrace{\hat \psi_n - \mathbb{E}[\hat \psi_n^S(\theta_0)|u_1,\dots,u_{n}]}_{\text{independent non-identically distributed}} + \underbrace{\hat \psi_n^S(\theta_0) - \mathbb{E}[\hat \psi_n^S(\theta_0)|u_1,\dots,u_n]}_{\text{scrambled sequence}}. \]\
      For the second term, Theorem \ref{th:owen} can be applied given that $\mathbb{E}[\hat \psi_i^S(\theta_0)|u_i]$ has finite variance. For the first term,
      Assumption \ref{ass:smooth} i.  implies a Lyapunov condition holds. As a result, the CLT for independent non-identically distributed arrays can be applied \citep[][Theorem 5.10]{White1984}.    
      Note that similar arguments implies that for each $\theta \in \Theta$, $(\hat \psi_n^S(\theta) - \mathbb{E}[\hat \psi_n^S(\theta)])= O_p(n^{-1/2})$, \textit{i.e.} pointwise convergence holds.\\

      \noindent \textbf{Part ii.} Stochastic Equicontinuity Result for $\hat \psi_n^S(\theta) - \mathbb{E}[\hat \psi_n^S(\theta)|u_1,\dots,u_n]$\\
      As in Part i., Assumption \ref{ass:smooth} i.  implies a Lyapunov condition holds for the envelope $\bar \psi$. This implies a Lindeberg condition for the envelope holds. Further, Assumption \ref{ass:smooth} iii. implies that:
      \begin{align*}
        \sup_{\|\theta_1-\theta_2\| \leq \delta_n} &\frac{1}{n} \sum_{i=1}^n \mathbb{E} \left[ \|[\tilde \psi(x_i,u_i;\theta_1) -\tilde \psi(x_i,u_i;\theta_2)]-\mathbb{E}[\tilde \psi(x_i,u_i;\theta_1) -\tilde \psi(x_i,u_i;\theta_2)|u_i] \|^2 | u_i \right]\\ &\leq  \left( \frac{1}{n} \sum_{i=1}^n \tilde C_1(u_i)^2 + \mathbb{E}[\tilde C_1(u_i)^2] \right) \times \delta_n^2\\ 
        &= \left( 2\mathbb{E}[\tilde C_1(u_i)^2] + o_p(1) \right) \times \delta_n^2,
      \end{align*}
      which goes to $0$ for all sequences $\delta_n \to 0$. The last equality comes from applying Theorem \ref{th:owen} to $\tilde C(u_i)^2$ which has finite variance by assumption. $\Theta$ is a compact and convex subset of $\mathbb{R}^{d_\theta}$ which is finite dimensional. Given the Lindeberg condition, pointwise convergence in Part i. and the $L^2$-smoothness result above holds, the Jain-Markus Theorem can be applied\footnote{See Example 2.11.13 and Theorem 2.11.9 in \citet{VanderVaart1996}.} which implies the desired stochastic equicontinuity result.\\

      \noindent \textbf{Part iii.} Taylor Expansion of $\mathbb{E}[\hat \psi_n^S(\theta)|u_1,\dots,u_n]$\\
      For all $\theta_1,\theta_2$, Assumption \ref{ass:smooth} iv. implies:
      \begin{align*}
        \|\frac{1}{n} &\sum_{i=1}^n \{ \mathbb{E}[\tilde \psi(x_i,u_i;\theta_1)-\tilde \psi(x_i,u_i;\theta_2)|u_i] - \partial_\theta \mathbb{E} \left[ \tilde \psi(x_i,u_i;\theta_2) | u_i \right](\theta_1-\theta_2) \} \|\\
        &\leq \frac{1}{n} \sum_{i=1}^n \tilde C_3(u_i) \times \| \theta_1-\theta_2 \|^2\\
        &= (\mathbb{E}[\tilde C_3(u_i)]+o_p(1))  \times \| \theta_1-\theta_2 \|^2,
     \end{align*}
     which implies the desired result. The last equality follows from Theorem \ref{th:owen} applied to $\tilde C_3(u_i)$ which has finite variance. Also note, that the conditions imply that the ULLN of Lemma \ref{lem:ULLNnocovariates} applies to $\partial_\theta \mathbb{E}[\hat \psi_n^S(\theta)|u_1,\dots,u_n]$ so that $\partial_\theta \mathbb{E}[\hat \psi_n^S(\theta)|u_1,\dots,u_n] = \partial_\theta \mathbb{E}[\hat \psi_n^S(\theta)] + o_p(1)$ uniformly in $\theta \in \Theta$. This concludes the proof.
    \end{proof}

    \begin{proof}[Proof of Proposition \ref{prop:static_covariates}] 
      By Lemma \ref{lem:se_clt_covariates}, $\hat \psi_n - \hat \psi_n^S(\theta)$ is stochastically equicontinuous which, together with Assumption \ref{ass:regular}, implies that $\hat \theta_n^S \overset{p}{\to} \theta_0$ by Theorem 2.1 in \citet{Newey1994a}. Then, using Lemma \ref{lem:se_clt_covariates} and standard arguments, we have:
      \begin{align*}
        0 &= G^\prime W \mathbb{E}\left[ \hat \psi_n - \hat \psi_n^S(\theta_0) \right]\\
        &= G^\prime W \left( \mathbb{E}\left[ \hat \psi_n - \hat \psi_n^S(\theta_0) \right] - \left[ \hat \psi_n - \hat \psi_n^S(\theta_0) \right] + \left[ \hat \psi_n - \hat \psi_n^S(\theta_0) \right] \right)\\
        &= G^\prime W \left( \mathbb{E}\left[ \hat \psi_n - \hat \psi_n^S(\hat \theta_n^S) |u_1,\dots,u_n \right] - \left[ \hat \psi_n - \hat \psi_n^S(\hat \theta_n^S) \right] + \left[ \hat \psi_n - \hat \psi_n^S(\theta_0) \right] \right) + o_p(n^{-1/2})\\
        &= G^\prime W \left( \mathbb{E}\left[ \hat \psi_n - \hat \psi_n^S(\hat \theta_n^S) |u_1,\dots,u_n \right]  + \left[ \hat \psi_n - \hat \psi_n^S(\theta_0) \right] \right) + o_p(n^{-1/2})\\
        &= G^\prime W \left( \mathbb{E}\left[ \hat \psi_n^S(\theta_0) - \hat \psi_n^S(\hat \theta_n^S) |u_1,\dots,u_n \right]  + \left[ \hat \psi_n - \hat \psi_n^S(\theta_0) \right] \right) + o_p(n^{-1/2}).
      \end{align*}
      The stochastic equicontinuity result can then be applied:
      \begin{align*}
        & \mathbb{E}\left[ \hat \psi_n - \hat \psi_n^S(\theta_0) |u_1,\dots,u_n \right] - \left[ \hat \psi_n - \hat \psi_n^S(\theta_0) \right]\\ &= \mathbb{E}\left[ \hat \psi_n - \hat \psi_n^S(\hat \theta_n^S) |u_1,\dots,u_n  \right] - \left[ \hat \psi_n - \hat \psi_n^S(\hat \theta_n^S) \right] +o_p(n^{-1/2}).
      \end{align*}
      Then, by Theorem \ref{th:owen}, $\|\mathbb{E}\left[ \hat \psi_n - \hat \psi_n^S(\theta_0) |u_1,\dots,u_n \right] - \mathbb{E}\left[ \hat \psi_n - \hat \psi_n^S(\theta_0)\right]\| = o_p(n^{-1/2})$ which allows to substitute $ \mathbb{E}\left[ \hat \psi_n - \hat \psi_n^S(\theta_0) \right]$ with the desired quantity.
      Using the CLT and stochastic equicontinuity result in Lemma \ref{lem:se_clt_covariates}:
      \begin{align*}
        \sqrt{n} \left(  \hat \theta_n^S -\theta \right) &= - \left( G^\prime W G \right)^{-1} G^\prime W \sqrt{n}[\hat \psi_n - \hat \psi_n^S(\theta_0)] + o_p(1)\\
        &\overset{d}{\to} \mathcal{N}(0,\Sigma),
      \end{align*}
      where
      \[ \Sigma = \left( G^\prime W G \right)^{-1} G^\prime W \tilde V W G \left( G^\prime W G \right)^{-1},\]
      $G = \partial_\theta \mathbb{E} [\hat \psi_n^S(\theta_0)]$, $\tilde V  = \mathbb{E} \left[ \text{var} \left( \psi(y_i,x_i) - \tilde \psi(x_i,u_i;\theta_0)|u_i  \right) \right]$.

      The results above are given for $S=1$. For $S>1$ fixed and finite, the simulated moments $\hat \psi_n^s$ are \textit{iid} over $s =1,\dots,S$. This implies that the CLT and stochastic equicontinuity results can be applied to each $s \in \{1,\dots,S\}$ and also apply to their average $\hat\psi_n^S$ by independence with $S$ fixed and finite. The remainder of the proof is identical which concludes the proof.
    \end{proof}

    \subsubsection{Scrambled Indirect Inference}

    \begin{proof}[Proof of Proposition \ref{prop:ind}] Assumption \ref{ass:ind} ii. implies a ULLN for $M_n^S(\theta;\psi)$ in $\psi$ for all $\theta \in \Theta$, by Lemma \ref{lem:ULLNnocovariates}. Then Assumption \ref{ass:ind} i. implies that Theorem 2.1 in \citet{Newey1994a} applies for each $\theta \in \Theta$ to $\hat \psi_n^S$ so that $\hat \psi_n^S(\theta)-\psi_\infty(\theta) = o_p(1)$ pointwise in $\theta \in \Theta$.

    Now, to prove that $\hat \theta_n^S$ itself is consistent, a ULLN for $\hat \psi_n^S$ in $\theta$ is needed. Given the pointwise consistency above, it remains to show that $\hat \psi_n^S$ is Lispchitz-continuous with stochastically bounded Lipschitz constant. For all $\theta_1,\theta_2 \in \Theta$, the mean-value theorem and the triangular inequality imply:
    \[ \|\hat \psi_n^S(\theta_1)-\hat \psi_n^S(\theta_2)\| \leq \|\partial_\theta \hat \psi_n^S(\tilde \theta)\| \times \|\theta_1-\theta_2 \|, \]
    where $\tilde \theta$ is some intermediate value. The implicit function theorem provides a closed-form for $\partial_\theta \hat \psi_n^S$ evaluated at any $\theta \in \Theta$:
      \[ \partial_\theta \hat \psi_n^S(\theta) = -\left[ \partial^2_{\psi,\psi^\prime} M_n^S(\theta;\hat \psi_n^S(\theta)) \right]^{-1} \partial^2_{\psi,\theta^\prime} M_n^S(\theta;\hat \psi_n^S(\theta)). \]
      Both $\partial^2_{\psi,\psi^\prime} M_n^S(\theta;\psi)$ and $\partial^2_{\psi,\theta^\prime} M_n^S(\theta;\psi)$ satisfy a ULLN in $(\theta,\psi)$ by Assumption \ref{ass:ind} and Lemma \ref{lem:ULLNnocovariates}. The Continuous Mapping Theorem then implies that $\partial_\theta \hat \psi_n^S(\theta) \overset{p}{\to} \mathbb{E}(\partial_\theta \hat \psi_n^S(\theta))$ pointwise in $\theta \in \Theta$.  Furthermore, Assumption \ref{ass:ind} iv. implies that:
      \[ \Big\|\left( \partial^2_{\psi,\psi^\prime} \mathbb{E}[M_n^S(\theta;\psi)] \right)^{-1} \partial^2_{\psi,\theta^\prime} \mathbb{E}[M_n^S(\theta;\psi)] \Big\| \leq \bar{M} < \infty,\]
      uniformly in $(\theta,\psi)$ for some finite bound $\bar{M} \geq 0$. Putting everything together, we have:
      \begin{align*}
        \|\hat \psi_n^S(\theta_1)-\hat \psi_n^S(\theta_2)\| &\leq \|\partial_\theta \hat \psi_n^S(\tilde \theta)\| \times \|\theta_1-\theta_2 \|\\
        &= \| \mathbb{E} [\partial_\theta \hat \psi_n^S(\tilde \theta)] + o_p(1)\| \times \|\theta_1-\theta_2 \|\\
        &\leq [\bar{M} + o_p(1)] \times \|\theta_1-\theta_2 \|
      \end{align*}
      This implies, as in \citet{Jennrich1969} and Proposition \ref{prop:static_nocovariates}, a ULLN for $\hat \psi_n^S$ over $\theta \in \Theta$.
  
      To establish the asymptotic normality for $\hat \theta_n^S$, first note that the ULLN for $\hat \psi_n^S$, $\partial^2_{\psi,\psi^\prime} M_n^S$ and $\partial^2_{\psi,\theta^\prime} M_n^S$ together with the implicit function theorem and the Lipschitz conditions imply a ULLN for $\partial_\theta \hat \psi_n^S$ in $\theta$.\footnote{The proof is omitted for brevity but is similar to the previous ULLNs.} By the usual mean-value expansion argument, this implies that:
      \[ \sqrt{n}[\hat \theta_n^S - \theta_0] = - \sqrt{n}\left[ \partial_\theta \psi_\infty (\theta_0)^\prime W\partial_\theta \psi_\infty (\theta_0)  + o_p(1)  \right]^{-1}  \partial_\theta \psi_\infty (\theta_0)^\prime W[ \hat \psi_n - \hat \psi_n^S(\theta_0)] + o_p(1). \]
      To conclude the proof, we need to show that $\sqrt{n}[\hat \psi_n^S(\theta_0)-\psi_\infty(\theta_0)] = o_p(1)$. Since $\hat \psi_n^S(\theta_0)$ is an M-estimator with the appropriate regularity conditions, the following holds:\footnote{The proof is very similar to Proposition \ref{prop:static_nocovariates}.}
      \[ \sqrt{n}[\hat \psi_n^S(\theta_0)-\psi_\infty(\theta_0)] = - \left[ \partial^2_{\psi,\psi^\prime} \mathbb{E}[M_n^S(\theta_0;\psi_\infty(\theta_0))] + o_p(1) \right]^{-1} \partial_\psi M_n^S(\theta_0;\psi_\infty(\theta_0)).  \]
      Since $\psi_\infty(\theta_0)$ is the population minimizer of $\mathbb{E}[ M_n^S(\theta_0;\cdot)]$, we have $\partial_\psi \mathbb{E}[M_n^S(\theta_0;\psi_\infty(\theta_0))]=0$. Applying Theorem \ref{th:owen} with Assumption \ref{ass:ind} iii. implies $\partial_\psi M_n^S(\theta_0;\psi_\infty(\theta_0)) = o_p(n^{-1/2})$ which, in turn, implies the desired result and concludes the proof.
      \end{proof}

      \subsection{Dynamic Models}

      \subsubsection{qMC-only Estimator}
    \begin{proof}[Proof of Proposition \ref{prop:dyn_case1}] Given the construction in Algorithm \ref{algo:ScrMMdyn1} and the assumptions the results of Lemma \ref{lem:ULLNnocovariates} hold and the proof proceeds as in Proposition \ref{prop:static_nocovariates}. This concludes the proof.
    \end{proof}

    \subsubsection{Hybrid MC-qMC Estimator}

    \begin{lemma}[Uniform Law of Large Numbers and CLT - MC-qMC] \label{lem:ULLN_mcQmc}
      Suppose that the Assumptions \ref{ass:regular}, \ref{ass:dyn_case2} hold then:
      \begin{itemize}
        \item[i.] $\sup_{\theta \in \Theta} \|\hat \psi_T^S(\theta) - \mathbb{E}[\hat \psi_T^S(\theta)]\| = o_p(1)$,
        \item[ii.] $\sup_{\theta \in \Theta} \| \partial_\theta \hat \psi_T^S(\theta) - \partial_\theta \mathbb{E}[\hat \psi_T^S(\theta)]\| = o_p(1)$,
        \item[iii.] $\sqrt{TS} \left( \hat \psi_T^S(\theta_0)-\mathbb{E}[\hat \psi_T^S(\theta_0)]  \right) \overset{d}{\to} \mathcal{N}(0,V)$ where $V = \lim_{T \to \infty} T \times var[ \hat \psi_T^S(\theta_0)|u_1,\dots,u_{TS}]$. 
      \end{itemize}
    \end{lemma}

    \begin{proof}[Proof of Lemma \ref{lem:ULLN_mcQmc}] $\,$ \newline
      \noindent \textbf{Part i.} ULLN for $\hat \psi_T^S(\theta)$\\
      The main steps are similar to Lemma \ref{lem:ULLNnocovariates} using pointwise convergence and Lipschitz continuity arguments.       
      The main difficulty is the presence of the Monte-Carlo terms $y_t^1,z_t^1$ which are dependent and non-stationary. To handle these, as in the proof of Lemma \ref{lem:se_clt_covariates},  separate $\hat \psi_T^S - \mathbb{E}[\hat \psi_T^S]$ into two components $(\hat \psi_T^S - \mathbb{E}[\hat \psi_T^S | u_1,\dots,u_{TS}])$ and $(\mathbb{E}[\hat \psi_T^S | u_1,\dots,u_{TS}] - \mathbb{E}[\hat \psi_T^S])$ to study the two individually.
      For the first term, \citet{Davydov1968}'s inequality implies pointwise convergence under mixing and moment conditions. For the second term, the non-stationarity implies that Theorem \ref{th:owen} does not apply directly. The geometric ergodicity conditions will allow to return to a setting where Theorem \ref{th:owen} applies. 

      As discussed above, for any $\theta \in \Theta$:
      \begin{align*}
        \hat \psi_T^S(\theta) - \mathbb{E}[\hat \psi_T^S(\theta)] &= \underbrace{\hat \psi_T^S(\theta) - \mathbb{E}[\hat \psi_T^S(\theta)|u_1,\dots,u_{TS}]}_{\text{heterogeneous dependent vector}} + \underbrace{\mathbb{E}[\hat \psi_T^S(\theta)|u_1,\dots,u_{TS}] - \mathbb{E}[\hat \psi_T^S(\theta)]}_{\text{non-stationary qMC sequence}}.
      \end{align*}
      For the first term, Davydov's inequality implies, up to a universal constant:
      \begin{align*}
        &\mathbb{E}[\| \hat \psi_T^S(\theta) - \mathbb{E}[\hat \psi_T^S(\theta)|u_1,\dots,u_{TS}] \|^2 | u_1,\dots,u_{TS}]\\ 
        &\leq \frac{1}{[TS]^2} \sum_{t=1}^{TS} \mathbb{E}[ \|\tilde \psi( y_t^1, z_t^1,u_t;\theta)-\mathbb{E}[\tilde \psi( y_t^1, z_t^1,u_t;\theta)|u_t]\|^2 |u_t]\\ &+ \frac{1}{[TS]^2} \sum_{t \neq t^\prime} \alpha(|t-t^\prime|)^{1/2} \times \mathbb{E}[ \| \tilde \psi( y_t^1, z_t^1,u_t;\theta)-\mathbb{E}[\tilde \psi( y_t^1, z_t^1,u_t;\theta)|u_t]\|^{4}|u_t ]^{1/4} \\ &\quad\quad\quad\quad\quad\quad\quad\quad\quad\quad\quad\times \mathbb{E}[ \| \tilde \psi(\tilde y_{t^\prime},\tilde z_{t^\prime},u_{t^\prime};\theta)-\mathbb{E}[\tilde \psi( y_{t^\prime}^1, z_{t^\prime}^1,u_{t^\prime};\theta)|u_{t^\prime}]\|^{4}|u_{t^\prime} ]^{1/4}.
      \end{align*}
      Note that $\mathbb{E}[ \|\tilde \psi( y_t^1, z_t^1,u_t;\theta)-\mathbb{E}[\tilde \psi( y_t^1, z_t^1,u_t;\theta)|u_t]\|^2 |u_t]$ is not stationary, so that Theorem \ref{th:owen} does not apply directly. However, by geometric ergodicity we have for any function $g$ with bounded fourth moment:
      \begin{align*}
        \| &\mathbb{E}[ g( y^1_t, z^1_t,u_t;\theta)-g( y^\infty_t, z^\infty_t,u_t;\theta) |u_t]\|\\ &= \|\int g( y^1, z^1,u_t;\theta)[f_t( y^1, z^1)-f_{\infty}( y^1, z^1)] d y^1 d z^1  \| \\
        &\leq \int \|g( y^1, z^1,u_t;\theta)\| \times |f_t( y^1, z^1)-f_{\infty}( y^1, z^1)| dy^1 d z^1 \\
        &\leq \left( \int \|g( y^1, z^1,u_t;\theta)\|^2 \times |f_t( y^1, z^1)-f_{\infty}( y^1, z^1)| d y^1 d z^1 \right)^{1/2} \left( \int |f_t( y^1, z^1)-f_{\infty}( y^1, z^1)| d y^1 d z^1 \right)^{1/2}\\
        &\leq \sqrt{2} \times \bar K_g \times \|f_t-f_\infty\|_{TV}^{1/2}\\
        &\leq \sqrt{2 C_1} \times \bar K_g \times \rho^{t/2},
      \end{align*}
      where $\bar K_g \geq 0$ is a bound for the moment conditional on $u_t$ fixed. This bound is finite by Assumption \ref{ass:dyn_case2} iii. and v. for $\tilde \psi$ and $\partial_\theta \tilde \psi$, respectively.
      Under the geometric ergodicity assumption, $\rho \in [0,1)$ so that $\sum_{t \geq 0} \rho^{t/2} <+\infty$ and:
      \begin{align*}
        &\frac{1}{[TS]^2} \sum_{t=1}^{TS} \mathbb{E}[ \|\tilde \psi( y_t^1, z_t^1,u_t;\theta)-\mathbb{E}[\tilde \psi( y_t^1, z_t^1,u_t;\theta)|u_t]\|^2 |u_t]\\ 
        &= \frac{1}{[TS]^2} \sum_{t=1}^{TS} \left( \mathbb{E}[ \|\tilde \psi( y^\infty_t, z^\infty_t,u_t;\theta)-\mathbb{E}[\tilde \psi ( y^\infty_t, z^\infty_t,u_t;\theta)|u_t]\|^2 |u_t]\right) + O(1/[TS]^2)\\
        &=  \mathbb{E} \left( \mathbb{E}[ \|\tilde \psi( y^\infty_t, z^\infty_t,u_t;\theta)-\mathbb{E}[\tilde \psi( y^\infty_t, z^\infty_t,u_t;\theta)|u_t]\|^2 ]\right)/[TS] + o_p(1/[TS]) + O(1/[TS]^2),
      \end{align*}
      where the last equality is due to Theorem \ref{th:owen} using the bounded fourth moment assumption to find the finite variance condition needed in the Theorem.

      The second term, which is a non-stationary qMC sequence, can be handled using the geometric ergodicity condition and the bounded fourth moment asusmption to get:
      \begin{align*}
        &\frac{1}{TS} \sum_{t=1}^{TS} \Big( \mathbb{E}[\tilde \psi( y_t^1, z_t^1,u_t;\theta)|u_t] - \mathbb{E}[\tilde \psi( y_t^1, z_t^1,u_t;\theta)] \Big)\\ &= \frac{1}{TS} \sum_{t=1}^{TS} \Big( \mathbb{E}[\tilde \psi( y^\infty_t, z^\infty_t,u_t;\theta)|u_t] - \mathbb{E}[\tilde \psi( y^\infty_t, z^\infty_t,u_t;\theta)] \Big) + O(1/[TS])\\
        &= o_p(1/\sqrt{TS}) + O(1/[TS]).
      \end{align*}

      Finally, the geometric ergodicity imply that $( y_t^1, z_t^1)_{t \geq 1}$ is $\alpha$-mixing with exponential decay. This implies that $\frac{1}{[TS]^2} \sum_{t \neq t^\prime} \alpha(|t-t^\prime|) = O(1/[TS])$ where $\alpha$ are the $\alpha$-mixing coefficients. Furthermore, by assumption $\mathbb{E}[ \| \tilde \psi( y_t^1, z_t^1,u_t;\theta)-\mathbb{E}[\tilde \psi( y_t^1, z_t^1,u_t;\theta)|u_t]\|^{4}|u_t ]^{1/4}$ is bounded for all $t \geq 1$. Altogether, these imply that:
      \begin{align*}
        \hat \psi_T^S(\theta) - \mathbb{E}[\hat \psi_T^S(\theta)] &= \underbrace{\hat \psi_T^S(\theta) - \mathbb{E}[\hat \psi_T^S(\theta)|u_1,\dots,u_{TS}]}_{=O_p(1/\sqrt{TS})} + \underbrace{\mathbb{E}[\hat \psi_T^S(\theta)|u_1,\dots,u_{TS}] - \mathbb{E}[\hat \psi_T^S(\theta)]}_{=o_p(1/\sqrt{TS})}\\
        &= O_p(1/\sqrt{TS}),
      \end{align*}
      which implies pointwise convergence.

    As in Proposition \ref{prop:static_nocovariates}, take a cover $\{\theta_1,\dots,\theta_J\}$ of $\Theta$ and:
    \begin{align*}
      \sup_{\theta \in \Theta} &\|\hat \psi_T^S(\theta) - \mathbb{E}[\hat \psi_T^S(\theta)]\|\\
      &\leq \max_{j \in \{1,\dots,J\}} \|\hat \psi_T^S(\theta_j) - \mathbb{E}[\hat \psi_T^S(\theta_j)]\| + \sup_{\theta \in \Theta} \min_{j \in \{1,\dots,J\}} \Big\| [\hat \psi_T^S(\theta)-\hat \psi_T^S(\theta_j)] - \mathbb{E}[ \hat \psi_T^S(\theta) - \hat \psi_T^S(\theta_j)] \Big\|.
    \end{align*}
    The first term can be handled with the pointwise convergence result above. For the second term, note that:
    \[ \|\hat \psi_T^S(\theta)-\hat \psi_T^S(\theta_j)\| \leq \frac{1}{TS} \sum_{t=1}^{TS} C_2( y_t^1, z_t^1,u_t;\theta_j) \times \|\theta-\theta_j\|. \]
    It is sufficient to show that $\sum_{t=1}^{TS} C_2( y_t^1, z_t^1,u_t;\theta_j) /[TS]$ is a $O_p(1)$ for each $j \in \{1,\dots,J\}$. Since $C_2$ satisfies the conditions for the pointwise convergence derived above, using the same arguments as for $\hat \psi_T^S$ we have: \[\frac{1}{TS}\sum_{t=1}^{TS} C_2( y_t^1, z_t^1,u_t;\theta_j) \overset{p}{\to} \mathbb{E}[C_2( y^\infty_t, z^\infty_t,u_t;\theta_j)].\]
    As in the proof of Lemma \ref{lem:ULLNnocovariates}, for $J$ and $T$ large enough we have:
    \[ \sup_{\theta \in \Theta} \|\hat \psi_T^S(\theta) - \mathbb{E}[\hat \psi_T^S(\theta)]\| \leq \varepsilon,\]
    with probability going to $1$, which implies the desired result.\\

    \noindent \textbf{Part ii.} ULLN for $\partial_\theta \hat \psi_T^S$\\
    Given the stated assumptions, the same results as above apply to $\partial_\theta \hat \psi_T^S$ uniformly in $\theta \in \Theta$.\\

    \noindent \textbf{Part iii.} CLT for $\sqrt{TS} \left( \hat \psi_T^S(\theta_0) - \mathbb{E}[\hat \psi_T^S(\theta_0)] \right)$\\
    In part i., it was shown that $\mathbb{E}[\hat \psi_T^S(\theta_0)|u_1,\dots,u_{TS}] -\mathbb{E}[\hat \psi_T^S(\theta_0)]= o_p(1/\sqrt{TS})$. Then, the bounded fourth moment in Assumption \ref{ass:dyn_case2} ii., the mixing condition i. and the variance condition vi. imply that the CLT for heterogeneous dependent arrays \citep[][Theorem 5.20]{White1984} can be applied and:
    \[ \sqrt{TS} \left( \hat \psi_n^S(\theta_0) - \mathbb{E}[\hat \psi_T^S(\theta_0)|u_1,\dots,u_{TS}] \right) \overset{d}{\to} \mathcal{N}(0,V), \]
    where $V = \lim_{T \to \infty} T \times \text{var}[\hat \psi_T^S(\theta_0)|u_1,\dots,u_{TS}]$. This concludes the proof.
    \end{proof}

    \begin{proof}[Proof of Proposition \ref{prop:dyn_case2}] Given the assumptions, Lemma \ref{lem:ULLN_mcQmc} applies and the proof proceed as in Proposition \ref{prop:static_nocovariates}. This concludes the proof.
    \end{proof}
\end{appendices}
\end{document}